%% file: main.tex
\documentclass[journal,final]{IEEEtran}

\usepackage{graphics,epstopdf,wrapfig}
\usepackage{graphicx}
\usepackage{changes}
\usepackage{xcolor}
\usepackage{epsfig}
\usepackage{subfig}
\usepackage{amsmath}
\usepackage{mathrsfs}
\usepackage{epsfig}
\usepackage{multirow}
\usepackage{amssymb,latexsym,amsfonts,amsmath}
\usepackage{graphicx}
\usepackage{color}
\usepackage{epstopdf}
\usepackage{float}
\usepackage{tikz}
\usepackage{caption}
\usepackage{tabu} 
\usepackage{verbatim}
\usepackage{tikz-3dplot}
\usetikzlibrary{shapes.geometric}
\usepackage{environ}
\NewEnviron{myequation}{%
\begin{equation}
\scalebox{1}{$\BODY$}
\end{equation}
}

\tdplotsetmaincoords{70}{30}
\usepackage{algpseudocode}
\usepackage{mathtools}
\usepackage[lined,ruled,commentsnumbered]{algorithm2e}

\makeatletter
\newcommand{\removelatexerror}{\let\@latex@error\@gobble}
\makeatother
\usepackage{amsthm}
\usepackage{pgfplots}
\pgfdeclarelayer{background}
\pgfdeclarelayer{foreground}
\pgfsetlayers{background,main,foreground}
\usetikzlibrary{intersections}
\usetikzlibrary{patterns}
\usetikzlibrary{automata,positioning,shapes,arrows}

\newtheorem{thm}{Theorem}
\newtheorem{problem}{Problem}
\newtheorem{lemma}{Lemma}
\newtheorem{corollary}{Corollary}
\newtheorem{definition}{Definition}
 
\usepackage{tikz}
\usetikzlibrary{automata,positioning,shapes,arrows}
\usetikzlibrary{shapes,snakes}
\tikzstyle{block} = [draw, rectangle, minimum size=3em]
\tikzstyle{bigblock} = [draw, rectangle, minimum height=7em, minimum width=11em]

\usepackage{pgfplots}
\pgfplotsset{compat=newest}
\usepgfplotslibrary{groupplots}
\usepgfplotslibrary{dateplot}

\theoremstyle{remark}

\tikzset{>=latex}

\theoremstyle{exampstyle}

\title{\LARGE\bf Policy Synthesis for Switched Linear Systems with Markov Decision Process Switching}

\author{Bo Wu, Murat Cubuktepe, Franck Djeumou, Zhe Xu, and Ufuk Topcu
	\thanks{ Bo Wu, Murat Cubuktepe, Franck Djeumou, Zhe Xu, and Ufuk Topcu are with the Department of Aerospace Engineering
and Engineering Mechanics, and the Oden Institute for Computational
Engineering and Sciences, University of Texas, Austin, 201 E 24th
St, Austin, TX 78712. email: {\tt\small $\{$bwu3, mcubuktepe, fdjeumou, zhexu, utopcu$\}$@utexas.edu}}}

\begin{document}
\maketitle

\begin{abstract}
We study the synthesis of mode switching protocols for a class of discrete-time switched linear systems in which the mode jumps are governed by  Markov decision processes (MDPs). We call such systems MDP-JLS for brevity.
Each state of the MDP corresponds to a mode in the switched system. The probabilistic state transitions in the MDP represent the mode transitions. We focus on finding a policy that selects the switching actions at each mode such that the switched system that follows these actions is guaranteed to be stable. Given a policy in the MDP, the considered MDP-JLS reduces to a Markov jump linear system (MJLS). {We consider both mean-square stability and stability with probability one. For mean-square stability, we leverage existing stability conditions for MJLSs and propose efficient semidefinite programming formulations to  find a stabilizing policy in the MDP.  For stability with probability one, we derive new sufficient   conditions and compute a stabilizing policy using linear programming. We also extend the policy synthesis results to  MDP-JLS with uncertain mode transition probabilities. 
}
\end{abstract}
\begin{IEEEkeywords}
	switched systems, Markov decision processes, optimization
\end{IEEEkeywords}

\section{Introduction}
Switched linear systems \cite{liberzon2003switching,sun2006switched} which  consist of a set of modes with linear dynamics  and a switching logic that describes the evolution of the modes, 
have recently found a broad range of applications, e.g.,  in robotics \cite{zegers2018distributed,xucontroller}, wireless sensor networks \cite{wu2014stability,zhang2017energy,wu2017formal}, networked control systems \cite{cetinkaya2018analysis}, security and privacy \cite{wu2018privacy,wu2020privacy}.

The switching logic for switched systems can be autonomous or controlled \cite{zhang2008l_}. {The former may be  the result of the system's own characteristics or an influence of its  environment, and the latter may be due to the designer's deliberate intervention.} 

In this paper, we study a class of the switched linear systems, where the switching logic is characterized by   Markov decision processes (MDPs) \cite{puterman2014markov}. \added{We name such systems MDP-JLS for brevity.} The MDP in an MDP-JLS includes a set of states that correspond to the modes in a switched system, a set of actions, and a  transition relation that defines the probability of transiting from the current mode to the next under a particular action. The mode switching in this MDP captures both deliberate intervention through the action selection and the environment or system uncertainties through the corresponding probabilistic mode transitions. Given a policy that selects the switching actions at each mode, an MDP-JLS reduces to a Markov jump linear system (MJLS) \cite{costa2006discrete} , where the mode switches in the system follow a discrete-time Markov chain (DTMC). 


We are interested in synthesizing a stabilizing policy for an MDP-JLS and consider both mean-square stability and stability with probability one of the induced MJLS. 
 We first show that policies that deterministically select an action in each mode of the MDP are not sufficient to \added{stabilize an MDP-JLS.}

For mean-square stability,  we introduce two approaches  to compute the stabilizing policies  based on existing stability conditions for MJLSs \cite{costa2006discrete,bolzern2015positive,shi2015survey,saravanakumar2017stability}. 
The first approach provides a sufficient condition for  a  policy  to stabilize an MDP-JLS and \added{formulates the policy synthesis problem  as a semidefinite programming problem that results in a  simultaneous search for both a 
stabilizing policy and a diagonal Lyapunov function for each mode. 
In the  second  approach, we  partition the variables for the policy and Lyapunov functions into two groups on which we perform coordinate descent~\cite{razaviyayn2013unified,shen2017disciplined}.} 
We alternate between searching for candidate Lyapunov functions and searching for a policy that satisfies the stability conditions by solving a semidefinite programming problem while fixing the other block of variables. 

For stability with probability one, we find stabilizing polices based on new sufficient stability conditions. These conditions extend the average dwell-time constraints for the stability in  non-stochastic switched linear systems \cite{hespanha1999stability,zhang2008l_,zhao2011stability} to MJLSs. More precisely, comparing with the traditional average dwell-time constraints, which require that the average time interval between any two consecutive mode switching is above a certain threshold, \added{The proposed stability conditions establish a lower bound for the probability of mode switching. 
Such conditions translate into constraints on the stationary distribution of the  induced DTMC}, based on which we solve policy synthesis  efficiently as a linear programming problem. 
We additionally extend the policy synthesis to   MDP-JLSs with uncertain mode transition probabilities and  the optimization of the expected   state-dependent costs.


We illustrate the use of the proposed methodologies with  two examples. 
For mean-square stability, we show that the method based on coordinate descent outperforms the semidefinite programming method.  For stability with probability one, since mean-square stability implies stability with probability one, the coordinate descent method is also applicable but scales poorly with the dimension of the linear dynamical systems in the modes. {On the other hand, we observe in the numerical examples that the computation time for the linear programming approach is not sensitive to the dimension of the dynamics in the modes and thus is more scalable. In oru experiments, the computation time of the linear programming approach is only a fraction (as low as $0.02\%$) of that of the coordinate descent method.}

\noindent {\bf{Related work.}} 
Stability and stabilization are major concerns for switched linear systems and have been extensively studied in the literature. The most popular stability analysis approaches for systems with arbitrary switching include common and multiple Lyapunov functions \cite{narendra1994common, branicky1998multiple} and (average) dwell-time conditions \cite{morse1996supervisory,hespanha1999stability,zhang2008exponential}. However, in many practical systems, switching between modes is often constrained due to physical limitations, and one usually has control over switching \cite{weiss2007automata}. As a result,  it is of interest to synthesize controllers that regulate the mode switching to stabilize the system while respecting  constraints on switching. In \cite{wang2016stability}, the authors studied a switched linear autonomous system where a finite automaton generates the mode switching sequences. Such a model considers non-stochastic mode switches governed by a finite automaton while the model proposed in this paper considers probabilistic  mode switches governed by an MDP. {Furthermore,  the focus of \cite{wang2016stability}  is only stability analysis   rather than stabilization.  References \cite{jha2010synthesizing,koo2001mode,liu2013synthesis} considered  switching controller synthesis for safety,  reachability   and temporal logic specifications, though only for switched systems with non-stochastic switching.}

{Recall that, for a given policy, an MDP-JLS reduces to an MJLS. For mean-square stability of an MJLS, there exist conditions that are both necessary and sufficient or only sufficient\cite{costa2006discrete,zhang2008analysis}. 
There also exist results for MJLS stability analysis that involve average dwell-time \cite{bolzern2010markov,chen2012h}. However, the average dwell-time in those papers refers to how frequently the transition probabilities of the DTMC change. While in MDP-JLSs, mode transition probabilities do not change over time and the average dwell-time represents the probability of mode switching.} 

\noindent {\bf{Contributions.}} \added{The contributions of this paper are three-fold.}
\begin{itemize}
    \item {We propose MDP-JLS, a new modeling framework for a class of switched linear systems where MDPs describes the switching logic.}
    \item {To guarantee mean-square stability, we study policy synthesis problems for MDP-JLSs with existing stability conditions of MJLSs and formulate them as  optimization problems  based on semidefinite   programming.}
    \item {For stability with probability one, we derive new stability conditions  with constraints in the probability of mode switching. We also consider MDP-JLSs with uncertain mode  transition probabilities and  the optimization of the expected average cost incurred by the switching actions.}
\end{itemize}
We organize the rest of this paper  as follows. Section \ref{sec:preliminaries} introduces the modeling framework and necessary definitions. Section \ref{sec:Problem Formulation} formulates the policy synthesis problem. {Solutions are proposed in Section \ref{sec:Stability Guaranteed Policy Synthesis} and Section \ref{sec:policy synthesis for stability with probability one} for mean-square stability and stability with probability one, respectively. Section \ref{sec:Examples} provides two examples to show the validity of the proposed solutions and compare their performances.} Section \ref{sec:Conclusion} concludes the paper and discusses future directions.


\section{Preliminaries}\label{sec:preliminaries}
In this section, we describe preliminary notions and definitions used in the sequel.

\emph{Notations:} $|S|$ denotes the cardinality of a set $S$. Given a real matrix $A\in\mathbb{R}^{m\times n}$, $A'$ denotes its transpose. We use $\mathbf{1}$ to denote a column vector with all elements equal to $1$.  $||A||_\infty:=\max_{1\leq i\leq m}\sum_{j=1}^n|a_{ij}|$. If $m=n$, $\rho(A)$ represents the spectral radius of $A$, i.e., $\rho(A):=\max_i|\lambda_i|$ where $\lambda_i,i\in\{1,\ldots,n\}$ are eigenvalues of $A$. Furthermore, $A>0$ ($A\geq 0$) denotes that the matrix $A$ is positive definite (positive semidefinite). $E[.]$ stands for computing the expectation. $\otimes$ denotes the Kronecker product. 
For $A_i\in\mathbb{R}^{n\times n},i\in\{1,\ldots,N\}$,  $diag(A_i)\in\mathbb{R}^{Nn\times Nn}$ represents the block diagonal matrix formed with $A_i$ at the diagonal and zero anywhere else, i.e.
$$
diag(A_i):=\begin{bmatrix}
    A_1 & 0 & 0  \\
    0 & \ddots & 0\\
    0 & 0 & A_N
\end{bmatrix}.
$$

\subsection{Switched Linear Systems}
Mathematically, a discrete-time switched linear system is described by 
\begin{equation}\label{eq:switched systems}
    x(k+1)=A_{s_k}x(k),
\end{equation}
where $x(k)\in\mathbb{R}^n$ is the state vector, $A_{s_k}\in\mathbb{R}^{n\times n}$  implies a matrix $A\in \{A_1,\ldots,A_{|S|}\}$.  The linear dynamics of (\ref{eq:switched systems}) is given by matrices $A_i$ when $s_k=i$, i.e, the  mode that the system is in  at time $k$. 



\subsection{Switched Linear Systems with Markov Decision Process Switching}
The system in (\ref{eq:switched systems}) in its general form is a hybrid system where the mode switches could depend on both the continuous dynamics and discrete mode \cite{liberzon2003switching,lin2009stability}.  In this paper, we consider a class of switched linear systems, where the mode switches are governed by a Markov decision process (MDP) \cite{puterman2014markov} defined as follows.
\begin{definition}
	An MDP is a tuple $\mathcal{M}=(S,\hat{s},\Sigma,T)$ which includes a finite set $S$ of states,
 an initial state $\hat{s}$,  a finite set $\Sigma$ of actions. 
		$T:S\times \Sigma\times S\rightarrow [0,1]$ is the probabilistic transition function with $
		T(s,\sigma,s'):=p(s'|s,\sigma),   \text{ for}\; s,s'\in S \text{ and } \sigma\in \Sigma$. We denote the number of modes, i.e., $|S|$ as $N$. 
\end{definition}

	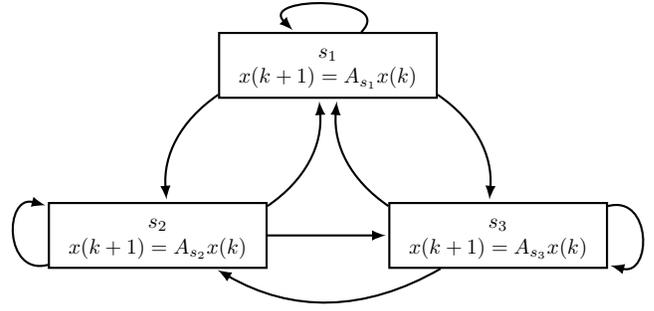
\begin{figure}[t!]
		\centering	
		\begin{tikzpicture}[shorten >=1pt,node distance=4cm,on grid,auto, thick,scale=0.8, every node/.style={transform shape}]
		\node[block] (q_1)   {\begin{tabular}{c}
		     $s_1$ \\ $x(k+1)=A_{s_1}x(k)$
		\end{tabular}};
		\node[block] (q_2) [below left = 4cm of q_1] {\begin{tabular}{c}
		     $s_2$ \\ $x(k+1)=A_{s_2}x(k)$
		\end{tabular}};
		\node[block] (q_3) [below right =4cm of q_1] {\begin{tabular}{c}
		     $s_3$ \\ $x(k+1)=A_{s_3}x(k)$
		\end{tabular}};

		\path[->]
		(q_1) edge [loop, above, looseness=2] node  {} (q_1)
		(q_1) edge [pos=0.5, bend right, above=0.5,sloped] node {} (q_2)
		(q_1) edge [pos=0.5, bend left, above=0.5,sloped] node {} (q_3)
		
		(q_2) edge [pos=0.5, loop left, looseness=2] node  {} (q_2)
		(q_2) edge [pos=0.5, bend right, above=0.5,sloped] node {} (q_1)
		(q_2) edge [pos=0.5, above=0.5] node {} (q_3)
		
		(q_3) edge [pos=0.5, loop right, looseness=2] node {} (q_3)
		(q_3) edge [pos=0.5, bend left, above=0.5,sloped] node {} (q_1)
		(q_3) edge [pos=0.5, bend left, above=0.5,sloped] node {} (q_2)
		;
		\end{tikzpicture}
		\caption{An MDP-JLS with three modes. Transition probabilities and actions are omitted.}\label{fig:mdp}
	\end{figure}
For simplicity, we denote $T_\sigma\in\mathbb{R}^{N\times N}$ as the transition probabilities induced by an action $\sigma\in\Sigma$ between state pairs. 
If $\sigma$ is not defined on a state $s_i$, $T(s_i,\sigma,s_j)=0$ for any $s_j\in S$. For notational simplicity, in the sequel we write $T(i,\sigma,j)$ instead of $T(s_i,\sigma,s_j)$. Then we formally define the system that we study in this paper as follows.

\begin{definition}
    An MDP-JLS is a switched system defined in (\ref{eq:switched systems}) with the mode switches governed by an MDP $\mathcal{M}=(S,\hat{s},\Sigma,T)$. 
\end{definition}

We shown an example of an MDP-JLS in Figure \ref{fig:mdp}.There are three linear dynamics corresponding to three modes, respectively, i.e.
$$
x(k+1)=A_{s_i}x(k),\text{ for } s_i\in S,
$$
\added{where} each state $s$, there is a set of actions available to choose. The nondeterminism of the action selection is resolved by a policy $\pi$.

\begin{definition}
	A (randomized) policy $\pi:S\times \Sigma \rightarrow [0,1]$  of an MDP $\mathcal{M}$ is a function that maps every state action pair $(s,\sigma)$ where $s\in S$ and $\sigma\in \Sigma$  with a probability $\pi(s,\sigma)$.    
\end{definition}
By definition, the policy $\pi$ specifies the probability to take an action $\sigma$  at a state $s$. For notational simplicity, we use $\pi(i,\sigma)$ for $\pi(s_i,\sigma)$.  As a result, given a policy $\pi$, the MDP $\mathcal{M}$ reduces to a discrete-time Markov chain (DTMC) $\mathcal{C}=(S,\hat{s},P)$. The matrix $P$ represents the transition probabilities and can be calculated by 
$$
P(s_i,s_j):=P_{ij}=\sum_{\sigma \in \Sigma} T(i,\sigma,j)\pi(i,\sigma),
$$
where $P_{ij}$ is the $(i,j)$ element of $P$ that denotes the probability of transitioning from $s_i$ to $s_j$ in one step. We can naturally generalize this definition to $n$-step transition probabilities where $P_{ij}^n$ denotes the probability of transitioning from $s_i$ to $s_j$ with $n$ steps, and $P_{ij}^n$ is the $(i,j)$ element of $P^n:=\prod_n P$. 

We denote the probability of being in a state $s_i$ at time step $n$ as $p^n(s_i)$ which can be computed by $P^n$ and the initial state. A stationary distribution (if it exists) over the states is a vector $\bar{p}\in\mathbb{R}^{|S|}$ that
$$
\bar{p}(s_i)=\sum_{s_j} \bar{p}(s_j)P_{ji}.
$$
In a DTMC $\mathcal{C}=(S,\hat{s},P)$, state $s_j$ is \emph{accessible} from $s_i$ if $P_{ij}^n>0$ for some non-negative $n<\infty$. Two states $s_i$ and $s_j$ are said to \emph{communicate} if $s_i$ is accessible from $s_j$ and $s_j$ is accessible from $s_i$. A state $s_i$ is said to be \emph{recurrent} if it is accessible from all states that are accessible from $s_i$. A state that is not recurrent is called \emph{transient}. A state $s_i$ is said to be aperiodic if $gcd(\{n|P^n_{ii}>0\})=1$ where $gcd$ stands for the greatest common divisor. A class $\mathcal{X}\subseteq S$ is a non-empty set of states where each $s_i\in\mathcal{X}$ communicates with every other state $s_j\in\mathcal{X}$ and communicates with no state $s_j\notin\mathcal{X}$. An \emph{ergodic class} is a class that consists of states that are both recurrent and aperiodic.
\begin{definition}
{A DTMC is an ergodic unichain if it contains a single ergodic class and maybe some transient states. An MDP is an ergodic unichain if every policy $\pi$ induces a DTMC that is an ergodic unichain. }
\end{definition}
{If a DTMC is an  ergodic unichain,   its stationary distribution $\bar{p}$ exists and is unique, and $\lim_{n\to\infty}p^n=\bar{p}$ \cite{gallager2013stochastic}. }

\subsection{Markov Jump Linear Systems}
The Markov jump linear system is defined as follows \cite{costa2006discrete}.
\begin{definition}
    A Markov jump linear system (MJLS) is a switched system defined in (\ref{eq:switched systems}) with the mode switches governed by a DTMC $\mathcal{C}=(S,\hat{s},P)$. 
\end{definition}

Given an MDP-JLS with an MDP $\mathcal{M}=(S,\hat{s},\Sigma,T)$ and a policy $\pi$, the resulting system is an MJLS whose mode switches  can be characterized by the DTMC $\mathcal{C}$ induced from the policy $\pi$.
If the system (\ref{eq:switched systems}) is in mode $s_i$, then the probability that it switches to mode $s_j$ is given by $P_{ij}$.

For MJLS analysis, stability is one of the major concerns. Several notions of stability exist in the literature \cite{shi2015survey}. In this paper, we are interested in mean-square stability and stability with probability one as defined below.

\begin{definition}\cite{costa2006discrete}
    An MJLS with the dynamics in \eqref{eq:switched systems} is said to be mean-square stable   if
    \begin{equation}\nonumber
    \begin{split}
    \lim_{k\rightarrow\infty}||E[x(k) x'(k)]||_\infty= 0  
    \end{split}
    \end{equation}
    for any initial condition $x_0$.
\end{definition}

\begin{definition}\cite{costa2006discrete}
    {An MJLS with the dynamics in \eqref{eq:switched systems} is said to be stable with probability one if
 $$
 \lim_{k\rightarrow\infty}||x(k)||_\infty= 0 \text{ with probability one}
$$
    for any initial condition $x_0$.}
\end{definition}
  {Mean-square stability is shown to imply stability with probability one \cite{costa2006discrete}. }

\subsection{Optimization basics}

In this paper, we use optimization problems such as linear programs (LPs), semidefinite programs (SDPs) and bilinear matrix inequalities (BMIs) extensively in the policy synthesis. We briefly define them as follows.

An LP is an optimization problem with a linear objective and constraints on the variable $y \in \mathbb{R}^n$, which is given by 

\begin{align}
    \text{minimize} &\quad c'y\label{eq:lp_objective}\\
  \text{subject to}&  \quad\nonumber\\
    &\quad Ay=b,\label{eq:lp_equality}\\
   &\quad \displaystyle x\geq 0,\label{eq:lp_positive}
\end{align}
where $A\in \mathbb{R}^{m \times n}$ is a given matrix, and $c \in \mathbb{R}^n, b \in \mathbb{R}^m$ are given vectors. The LP in this form is also referred to as the \emph{standard form}~\cite[Chapter 4.3]{Boyd}. LPs are convex optimization problems and can be solved efficiently using interior point methods~\cite{nesterov1994interior,Boyd}.

An SDP is an optimization problem with a linear objective, linear equality constraints and a matrix nonnegativity constraint on the variable $y \in \mathbb{R}^{n}$, which can be written as
\begin{align}
    \text{minimize} &\quad c'y\label{eq:sdp objective}\\
  \text{subject to}&  \quad\nonumber\\
    &\quad Ay=b,\label{eq:sdp_equality}\\
   &\quad \displaystyle \sum^n_{i=1} y_i F_i\geq F_0,\label{eq:sdp_positive}
\end{align}
where $F_0,\ldots,F_m \in \mathbb{R}^{p \times p}$, are given symmetric matrices, $A \in \mathbb{R}^{m \times n}$ is a given matrix, and $c \in \mathbb{R}^n, b \in \mathbb{R}^m$ are given vectors. SDPs are convex optimization problems, and is a generalization of LPs. SDPs can also be solved efficiently~\cite{nesterov1994interior,Boyd}. The constraint in~\eqref{eq:sdp_positive} is named as a linear matrix inequality (LMI), and it is a convex constraint in $y$.

A BMI can be written as the following form:
\begin{align*}
 &\quad \displaystyle \sum^n_{i=1} y_j F_i+\sum^m_{j=1} z_j G_j+\sum^n_{i=1}\sum^m_{j=1} y_i z_j H_{ij}\geq F_0,
\end{align*}
where $F_i, G_j, H_{ij} \in \mathbb{R}^{p \times p}$ for $i={1,\ldots,n}$ and $j={1,\ldots,m}$ are given symmetric matrices, and $x \in \mathbb{R}^n, y \in \mathbb{R}^m$ are a vector of variables. A BMI is an LMI in $y$ for fixed $z$ and an LMI in $z$ for fixed $y$. The bilinear terms in a BMI make the feasible set not jointly convex in $y$ and $z$ and it is generally hard to find a feasible solution to a BMI~\cite{vanantwerp2000tutorial}.

\section{Problem Formulation}\label{sec:Problem Formulation}
In traditional MDP literature, finding a policy for an optimized expected cost \cite{puterman2014markov} or to satisfy a specification in temporal logic \cite{baier2008principles} is of the primary concern. However, in this paper, we are interested in synthesizing a policy $\pi$ in an MDP that governs the switches of a dynamical system defined in (\ref{eq:switched systems}). In this case, the objective is to stabilize an MDP-JLS.         


\begin{problem}[Synthesis for  mean-square stable]\label{problem:mean-square stability}
\added{Given an MDP-JLS, find a policy $\pi:S\times \Sigma\rightarrow [0,1]$ such that the resulting MJLS is mean-square stable.}
\end{problem}

\begin{problem}[Synthesis for stability with probability one]\label{problem:stability with probability one}
 \added{ Given an MDP-JLS, find a policy $\pi:S\times \Sigma\rightarrow [0,1]$ such that the resulting MJLS is stable with probability one.}
\end{problem}

\section{mean-square Stability Guaranteed Policy Synthesis}\label{sec:Stability Guaranteed Policy Synthesis}
This section solves Problem \ref{problem:mean-square stability}. Since an MDP-JLS will reduce to an MJLS with a policy, we first review some stability conditions in MJLS that we will leverage to synthesize policies. 
\subsection{Stability Conditions}
We give two necessary and sufficient stability conditions for an MJLS as the following.
\begin{thm}\cite{costa2006discrete}\label{thm:necceary and sufficient}
Given an MJLS as defined in (\ref{eq:switched systems}) whose mode $s\in S$ makes random transitions described by a DTMC $\mathcal{C}=(S,\hat{s},P)$,  the following assertions are equivalent.


\begin{enumerate}
    \item The MJLS is mean-square stable.
    \item $\rho(\mathcal{A})< 1,$ where
    $$
    \mathcal{A} = (P'\otimes I)diag(A_i\otimes A_i),
    $$
    and  $I$ is the identity matrix of a proper dimension.
    \item  There exists a $V=(V_1,\ldots,V_N)\in \mathbb{R}^{n\times n}$ with $V>0$  such that  
	\begin{equation}\label{eq:n&s condition}
	   V-\mathcal{T}(V)>0, 
	\end{equation}
	where 
	$$
	\mathcal{T}_j(V)=\sum_{i=1}^N P_{ij}A_iV_iA'_i.
	$$
\end{enumerate}
\end{thm} 

Note that the stability conditions do not depend on either the initial state $\hat{s}$ of the MDP or the initial continuous state $x(0)$.  For computational efficiency, we state a sufficient stability condition as follows.  
\begin{corollary}\cite{costa2006discrete}\label{corollary:sufficient}
Given an MJLS as defined in (\ref{eq:switched systems}) whose mode $s\in S$ makes transitions following a DTMC $\mathcal{C}=(S,\hat{s},P)$,  the MJLS is mean-square stable if there exists $\alpha_i>0$ such that the following is satisfied.

	\begin{equation}\label{eq:sufficient condition 1}
\alpha_iI-\sum_{j=1}^N P_{ij}\alpha_jA_iA'_i>0,\; i \in \lbrace{1,\ldots,N\rbrace}.
\end{equation}
\end{corollary}
The condition given in~\eqref{eq:n&s condition} can be checked by solving an SDP with $V_i$ as variables. However, the number of variables for this SDP is $n^2\cdot N$, and finding a feasible solution for the SDP can be time consuming for large $n$ and $N$. On the other hand, the condition in~\eqref{eq:sufficient condition 1} can be checked by solving an SDP with $N$ variables, and the size of the optimization problem is smaller compared to the optimization problem in~\eqref{eq:n&s condition}. 


\subsection{\added{Deterministic Policies Are Not Sufficient for Stability}}
We first show that a deterministic policy, i.e, $\pi:S\rightarrow A$ is not sufficient to guarantee the  stability of an MDP-JLS. It means that there may not exist a deterministic policy to stabilize an MDP-JLS, but there exists an randomized policy that achieves stability. 


We illustrate this fact by a counterexample. Consider a switched system with system dynamics in (\ref{eq:switched systems})
$$A_1=
\begin{bmatrix}
    0.99 & -0.56  \\
    -0.19 & 0.73
\end{bmatrix}
\text{ and }
A_2=
\begin{bmatrix}
    0.38 & -0.98  \\
    -0.66 & -0.66 
\end{bmatrix}.$$
The MDP $\mathcal{M}=(S,\hat{s},\Sigma,T)$ where $S=\{s_1,s_2\}$ and $\Sigma=\{\sigma_1,\sigma_2\}$. The transition probabilities induced by action $\sigma_1$ and $\sigma_2$ are
$$T_{\sigma_1}=
\begin{bmatrix}
    0.21 & 0.79  \\
    0.90 & 0.10 
\end{bmatrix}
\text{ and }
T_{\sigma_2}=
\begin{bmatrix}
    0.71 & 0.29  \\
    0.13 & 0.87 
\end{bmatrix}.$$
 The deterministic policy that induces a minimal spectral radius  is selecting $\sigma_1$ in both mode $1$ and $2$. The spectral radius $\rho(\mathcal{A})$ of the MJLS induced by this policy is $1.04>1$, which makes the overall system unstable. However, the policy that selects $\sigma_1$ in mode 1, and selects $\sigma_1$ in mode 2 with a probability of $0.27$ induces an MJLS that has a spectral radius of $\rho(\mathcal{A})=0.90<1$. So the system is stable according to Theorem \ref{thm:necceary and sufficient}.  Therefore, we conclude that deterministic policies are not sufficient to stabilize an MDP-JLS.


\subsection{Policy Synthesis via Bilinear Matrix Inequalities}\label{subsec:sos}
In this section, we formulate a condition based on bilinear matrix inequalities to synthesize a randomized policy to stabilize an MDP-JLS. The condition is a straightforward generalization of the linear matrix inequalities given in~\eqref{eq:n&s condition}. The following result states that we can  search for a stabilizing policy by finding a solution to a set of bilinear matrix inequalities.

\begin{thm}\label{thm:bilinear}
Consider an MDP-JLS  whose mode $s\in S$ makes transitions following a MDP $\mathcal{M}=(S,\hat{s},\Sigma,T)$ and the dynamics in each modes as in \eqref{eq:switched systems}. If there exists matrices $V_i \in \mathbb{R}^{n \times n}$, and $\pi$ such that the following holds:
\begin{align}
  &  V_i >0,\label{eq:positive_sos}\\
   & 	   V-\mathcal{T}(V)>0,\label{eq:positive2_sos}\\
&	\mathcal{T}_j(V)=\sum_{i=1}^N  P_{ij}A_iV_iA_i^{'},\label{eq:positive3_sos}\\
& P_{ij}=\sum_{\sigma \in \Sigma} T(i,\sigma,j)\pi(i,\sigma),\label{eq:induced dtmc_sos}\\
& \sum_{\sigma \in \Sigma} \pi_{i,\sigma}=1,\label{eq:policy_welldefined_sos}\\
&\pi(i,\sigma)\geq 0,\label{eq:policy_nonnegative_sos}
\end{align}
for $i,j= \lbrace{1,\ldots,N\rbrace}$ and $\sigma \in \Sigma$, then the induced MJLS is mean-square stable.
\end{thm}

\begin{proof}
Constraints \eqref{eq:induced dtmc_sos}, \eqref{eq:policy_welldefined_sos}, \eqref{eq:policy_nonnegative_sos} construct the induced DTMC $\mathcal{C}$ with transitions governed by $P_{ij}.$ Using the result of Theorem 1, the constraints \eqref{eq:positive_sos}, \eqref{eq:positive2_sos}, and  \eqref{eq:positive3_sos} ensure that the MJLS is mean-square stable with the induced DTMC $\mathcal{C}$. Hence, the existence of a policy and matrices $V_i$ that satisfies the constraints \eqref{eq:positive_sos}--\eqref{eq:induced dtmc_sos} shows that the  MJLS is mean-square stable.
\end{proof}

Note that the constraints given in~\eqref{eq:positive_sos}--\eqref{eq:policy_nonnegative_sos} are BMI constraints due to multiplication between variables $\pi$ and $V$ in ~\eqref{eq:positive2_sos}--\eqref{eq:induced dtmc_sos}, therefore it is hard in general to find a policy by solving the BMI directly. In the next section, we propose two approaches based on convex optimization that are easier to compute, and discuss their relationship with the BMI in~\eqref{eq:positive_sos}--\eqref{eq:policy_nonnegative_sos}.

\subsection{Policy Synthesis via Convex Optimization}\label{subsec:co}
In this section, we propose two methods to synthesize a policy that stabilizes an MDP-JLS. The first method is based on checking feasibility of an SDP, which is an relaxation of the original stability condition. The second method is based on applying a coordinate descent on the variables $V$ and $\pi$. We can use coordinate descent in our case efficiently, as the constraints in~\eqref{eq:positive2_sos}--\eqref{eq:induced dtmc_sos} are LMI constraints if $V$ or $\pi$ is fixed.

\subsubsection{Semidefinite Relaxation}

In the following, we state the semidefinite relaxation to compute a policy that stabilizes an MDP-JLS. the relaxation extends the stability condition given in~\eqref{eq:sufficient condition 1} for an MJLS to a switched system whose mode switches are governed by an MDP. 

\begin{thm}\label{thm:sdp}
Consider an MDP-JLS  whose mode $s\in S$ makes transitions following a MDP $\mathcal{M}=(S,\hat{s},\Sigma,T)$ and the dynamics in each modes as in \eqref{eq:switched systems}. If there exists $K_{i,\sigma},\alpha_i \in \mathbb{R}>0$ such that
\begin{align}
  &  V_i=\alpha_i I >0,\label{eq:positive_sdp}\\
   & 	   V-\mathcal{T}(V)>0,\label{eq:positive2_sdp}\\
&	\mathcal{T}_j(V)=\sum_{i=1}^N \sum_{\sigma \in \Sigma} T(i,\sigma,j)K_{i,\sigma}A_iA'_i,\label{eq:positive3_sdp}\\
& \sum_{\sigma \in \Sigma} K_{i,\sigma}=\alpha_i,\label{eq:policy_welldefined_sdp}\\
&K_{i,\sigma}\geq 0,\label{eq:policy_nonnegative_sdp}
\end{align}
for $i,j=\lbrace{1,\ldots,N\rbrace}$ and $\sigma \in \Sigma$, then the MJLS is mean-square stable.
\end{thm}

\begin{proof}
Suppose that the condition given by constraints \eqref{eq:positive_sos}--\eqref{eq:policy_nonnegative_sos} is satisfied with $V_i=\alpha_i I>0, i=\lbrace{1,\ldots,N \rbrace}$. Then, the constraint~\eqref{eq:positive3_sos} becomes
\begin{align}
&	\mathcal{T}_j(V)=\sum_{i=1}^N \sum_{\sigma \in \Sigma} T(i,\sigma,j)\pi(i,\sigma)\alpha_i A_i A'_i\label{eq:positive3_bdp}
\end{align}
with variables $\alpha_i>0, i=\lbrace{1,\ldots,N \rbrace}$ and $\pi.$ Note that for a given policy $\pi$ and the induced DTMC $\mathcal{C}$, the constraint in~\eqref{eq:positive3_bdp} is equivalent to the condition given by~\eqref{eq:sufficient condition 1} in Corollary 1. By defining the change of variable $K_{i,\sigma}= \pi(i,\sigma)\cdot\alpha_i $ for $i=\lbrace{1,\ldots,N\rbrace}$ and $\sigma\in\Sigma$, the constraints~\eqref{eq:positive3_sos}--\eqref{eq:policy_nonnegative_sos} are equivalent to the constraints in ~\eqref{eq:positive3_sdp}--\eqref{eq:policy_nonnegative_sdp}. Finding a feasible solution that satisfies the constraints in \eqref{eq:positive_sos}--\eqref{eq:policy_nonnegative_sos} yields a policy $\pi(i,\sigma)=K_{i,\sigma}/\alpha_i$ for $i=\lbrace{1,\ldots,N\rbrace}$ and $\sigma \in \Sigma$, which by construction satisfies the constraints in \eqref{eq:positive_sos}--\eqref{eq:policy_nonnegative_sos}. Therefore, the policy $\pi$ and $V$ ensures that the induced MJLS is mean-square stable. 
\end{proof}

The constraints in \eqref{eq:positive_sdp}--\eqref{eq:policy_nonnegative_sdp} are LMIs in the variables $K$ and $\alpha$. Finding a feasible solution of a set of LMIs can be done by solving an SDP. However, this condition is only a sufficient as we restrict the structure of the  matrix $V$, therefore we may not be able to certify the stability of an MJLS even though there may exists a policy that ensures that the induced MJLS is mean-square stable.





\subsubsection{Coordinate Descent}

In this section we discuss the coordinate descent (CD) approach and the differences with a basic CD algorithm. Recall that a BMI is an LMI if one the variables is fixed, and we can check if the constraints in~\eqref{eq:positive_sos}--\eqref{eq:policy_nonnegative_sos} are feasible for a fixed $V$ or $\pi$. However, applying the basic CD on $V$ and $\pi$ requires the problems to be feasible for a fixed $V$ or $\pi$, which is not necessarily true in our case. If the initial problem is feasible, then we know that $\pi$ stabilizes the induced MJLS. Therefore, we assume that the initial policy does not stabilize the system.

Our implementation differs from a basic coordinate descent algorithm in the addition of the slack variables to the constraint in~\eqref{eq:positive2_sos}, which ensures that the resulting LMI is feasible for a fixed set of variables, and we use a proximal update between the variables instead of the original update method between $V$ and $\pi$. Details about the proximal update and the convergence guarantees can be found in~\cite{xu2013block}.

We start with an initial guess of the variables $V^0$ and $\pi^0$. Then in each iteration $k$, 
we solve the following SDP for a fixed $\pi^k$:
\begin{align}
\text{minimize}&\quad -\gamma+\displaystyle\sum^N_{i=1} L||V_i-V^{k-1}_i||_2\label{eq:CD}\\
\text{subject to}&\nonumber\\
  &  \quad V_i >0,\label{eq:positive_cdv}\\
   & \quad	   V-\mathcal{T}(V)\geq \gamma I,\label{eq:positive2_cdv}\\
&	\quad\mathcal{T}_j(V)=\sum_{i=1}^N  P_{ij}A_iV_iA'_i,\label{eq:positive3_cdv}\\
&\quad P_{ij}=\sum_{\sigma \in \Sigma} T(i,\sigma,j)\pi^k(i,\sigma),\label{eq:induced dtmc_cdv}
\end{align}
where $V_i\in \mathbb{R}^{n \times n}, i=\lbrace{1,\ldots,N\rbrace}$ and $\gamma \in \mathbb{R}$ are variables, and $L \in \mathbb{R}$ is a small positive constant. After we get $V$ from \eqref{eq:CD}--\eqref{eq:induced dtmc_cdv}, the SDP we solve for a fixed $V^k=V$ is given as follows:
\begin{align}
\text{minimize}&\quad -\gamma+\displaystyle\sum^N_{i=1}\sum_{\sigma \in \Sigma} L||\pi(i,\sigma)-\pi^{k-1}(i,\sigma)||_2\\
\text{subject to}&\nonumber\\
   & 	 \quad  V^k-\mathcal{T}(V^k)\geq \gamma I,\label{eq:positive2_cdp}\\
&	\quad\mathcal{T}_j(V)=\sum_{i=1}^N  P_{ij}A_iV^k_iA'_i,\label{eq:positive3_cdp}\\
&\quad P_{ij}=\sum_{\sigma \in \Sigma} T(i,\sigma,j)\pi(i,\sigma),\label{eq:induced dtmc_cdp}\\
& \quad\sum_{\sigma \in \Sigma} \pi_{i,\sigma}=1,\label{eq:policy_welldefined_cdp}\\
&\quad\pi(i,\sigma)\geq 0\label{eq:policy_nonnegative_cdp}
\end{align}
with variables $\pi$ for $i=\lbrace{1,\ldots,N\rbrace}$ and $\sigma \in \Sigma$, and $\gamma \in \mathbb{R}$. After solving each SDP, we update the variables until we converge to a solution or we obtain a solution with $\gamma>0$. If we can find a solution with $\gamma>0$, the conditions~\eqref{eq:positive2_cdv} and~\eqref{eq:positive2_cdp} implies the condition given in~\eqref{eq:positive2_sos}, and the rest of the conditions in~\eqref{eq:positive_sos}--\eqref{eq:policy_nonnegative_sos} are already satisfied in either SDPs that we solve during CD. In this case, we stop the algorithm as the solution given by $V$ and $\pi$ guarantees that the MJLS is mean-square stable. Note that our method is guaranteed to converge as we use the update (1.3b) in~\cite{xu2013block}, however the procedure can converge to a solution with $\gamma\leq 0$, which implies that the CD method cannot certify if the MJLS is mean-square stable.


\section{\added{Policy Synthesis for Stability with Probability One}}\label{sec:policy synthesis for stability with probability one}

This section solves Problem \ref{problem:stability with probability one}.  We assume that  $\mathcal{M}$ is an ergodic unichain MDP which has a unique stationary distribution for every DTMC induced by a policy $\pi$. For a stationary distribution $p^\infty$, the probability to jump to state $s$ from a different state in one time step denoted as $\overrightarrow{p_s}$  is given by 
\begin{equation}\label{eq:jump to state s}
   \overrightarrow{p_s} =  \sum_{s'\neq s}P(s',s)p^\infty_{s'}.
\end{equation}
Furthermore, the event that there is a mode jump is Bernoulli distributed such that
\begin{equation}
        P_{jump} = 1 - \sum_sp_s^\infty P(s,s)=\sum_s \overrightarrow{p_s},
\end{equation}
where $P_{jump}$ denotes the probability of a mode jump.
\subsection{Stability with Mode-Independent Conditions}
\begin{thm}\label{thm:main}
Consider an MDP-JLS  whose mode $s\in S$ makes transitions following a MDP $\mathcal{M}=(S,\hat{s},\Sigma,T)$ and the dynamics in each modes as in \eqref{eq:switched systems} with given constants $0<\alpha<1$ and $\mu>1$. Assume there exists a Lyapunov function candidate $V(x)=\{V_s(x),s\in S\}$ that satisfies the following for any pair of $s,s'\in S$. 
\begin{align}
   & V_{s}(x_{k+1})-V_{s}(x_{k})\leq -\alpha V_{s}(x_{k}),\label{eq:same mode decrease}\text{ and }\\
   & V_{s}(x_{k})\leq\mu V_{s'}(x_{k}), \label{eq: mode jump decrease}
\end{align}
where $x_{k}:=x(k)$ for simplicity. Given a policy $\sigma$ and if the induced DTMC has a stationary distribution, then the system is stable with probability one if
\begin{equation}
    P_{jump}< \frac{\ln(\frac{1}{1-\alpha})}{\ln(\mu)}.
\end{equation}
\end{thm}

\begin{proof}
For a time horizon $k$, suppose there are $m$ mode jumps so far at time instants $k_i,i\in\{1,...,m\}$ such that $s_{k_i}\neq s_{k_i-1}$, i.e., $k_i$ is the time instant that the $i$-th jump just occurred. Then we   have
\begin{equation}
    V_{s_{k}}(x_k)\leq (1-\alpha)^{k-k_{m}}V_{s_{k_{m}}}(x_{k_{m}})
\end{equation}
from equation (\ref{eq:same mode decrease}) since the mode remains the same during the time interval $[k_{m},k]$, i.e. $s_t = s_k$ for all $t\in[k_m,k]$. Together with equations \eqref{eq:same mode decrease} and  (\ref{eq: mode jump decrease}), we   have
\begin{equation}
    \begin{split}
       V_{s_{k}}(x_k)&\leq  (1-\alpha)^{k-k_{m}}\mu V_{s_{k_{m}-1}}(x_{k_{m}})\\
       &\leq(1-\alpha)^{k-k_{m}}\mu (1-\alpha)^{k_m-k_{m-1}}\mu V_{s_{k_{m-1}-1}}(x_{k_{m-1}}),\\
       &\leq \dots \leq (1-\alpha)^k\mu^m V_{\hat{s}}(x_0)\\&=((1-\alpha)\mu^{\frac{m}{k}})^kV_{\hat{s}}(x_0).
    \end{split}
\end{equation}
To prove the system is stable with probability one, we need to prove that 
$$        P(\lim_{k\rightarrow\infty}((1-\alpha)\mu^{\frac{m}{k}})^k=0)=1,
$$
which implies that
\begin{equation}\nonumber
    \begin{split}
    P(\lim_{k\rightarrow\infty}(1-\alpha)\mu^{\frac{m}{k}}< 1)&=1\\
        \Longleftrightarrow  P(\lim_{k\rightarrow\infty}{\frac{m}{k}}< \frac{\ln(\frac{1}{1-\alpha})}{\ln(\mu)})&=1.
            \end{split}
\end{equation}
Using the law  of total probability and conditional probability, we have that
\begin{equation}\label{eq:sum}
    \begin{split}
        &P(\lim_{k\rightarrow\infty}{\frac{m}{k}}< \frac{\ln(\frac{1}{1-\alpha})}{\ln(\mu)}|\mathcal{B})P(\mathcal{B})\\&+P(\lim_{k\rightarrow\infty}{\frac{m}{k}}< \frac{\ln(\frac{1}{1-\alpha})}{\ln(\mu)}|\bar{\mathcal{B}})P(\bar{\mathcal{B}})=1,
    \end{split}
\end{equation}
where   $\mathcal{B}$ represents the event that  $\lim_{k\rightarrow\infty}\frac{m}{k}=P_{jump}$ and   $\bar{\mathcal{B}}$ represents the event that  $\lim_{k\rightarrow\infty}\frac{m}{k}\neq P_{jump}$.
According to the law of large numbers \cite{bremaud2013markov}, we know that
$$
P(\mathcal{B})=P(\lim_{k\rightarrow\infty}\frac{m}{k}=P_{jump})=1 \text{ and }P(\bar{\mathcal{B}})=0.
$$
Therefore we have that
\begin{equation}
\begin{split}
& P(\lim_{k\rightarrow\infty}{\frac{m}{k}}< \frac{\ln(\frac{1}{1-\alpha})}{\ln(\mu)}|\bar{\mathcal{B}})P(\bar{\mathcal{B}}) =0.
\end{split}
\end{equation}
Thus from \eqref{eq:sum} we require 
\begin{equation}\label{eq:mode dependent thm 1}
\begin{split}
 & P(\lim_{k\rightarrow\infty}{\frac{m}{k}}< \frac{\ln(\frac{1}{1-\alpha})}{\ln(\mu)}|\mathcal{B})P(\mathcal{B})
 \\&= P(\lim_{k\rightarrow\infty}{\frac{m}{k}}< \frac{\ln(\frac{1}{1-\alpha})}{\ln(\mu)}|\lim_{k\rightarrow\infty}\frac{m}{k}=P_{jump})\\
 &=  P[P_{jump}< \frac{\ln(\frac{1}{1-\alpha})}{\ln(\mu)}]=1.
\end{split}
\end{equation}
Since given a policy $\pi$, $P_{jump}$ is a constant, equivalently we require that
$$
P_{jump}< \frac{\ln(\frac{1}{1-\alpha})}{\ln(\mu)}. 
$$
\end{proof}

\subsection{Stability with Mode-Dependent Conditions}
From Theorem \ref{thm:main}, we find a sufficient condition for the policy of the MDP to satisfy, such that the stability can be guaranteed. However, the conditions \eqref{eq:same mode decrease} and \eqref{eq: mode jump decrease} in Theorem \ref{thm:main} are mode-independent, which may introduce conservativeness because the same pair of $\alpha$ and $\mu$ has to be satisfied by all the modes. Therefore, inspired by \cite{zhao2011stability} we introduce the following theorem where the parameters such as $\alpha$ and $\mu$ in \eqref{eq:same mode decrease} and \eqref{eq: mode jump decrease} are mode-dependent.
\begin{thm}\label{thm:main2}
Consider an MDP-JLS  whose mode $s\in S$ makes transitions following a MDP $\mathcal{M}=(S,\hat{s},\Sigma,T)$ and the dynamics in each modes as in \eqref{eq:switched systems} with given constants $0<\alpha_s<1$ and $\mu_s>1$ for all $s\in S$. Assume there exists a Lyapunov function candidate $V(x)=\{V_s(x),s\in S\}$ that satisfies the following for any pair of $s,s'\in S$. 
\begin{align}
   & V_{s}(x_{k+1})-V_{s}(x_{k})\leq -\alpha_s V_{s}(x_{k}), \text{ and }\label{eq:same mode decrease 1}\\
   & V_{s}(x_{k})\leq\mu_s V_{s'}(x_{k}), \label{eq: mode jump decrease 1}
\end{align}
where $x_{k}:=x(k)$ for simplicity. Given a policy $\sigma$ and if the induced DTMC has a unique stationary distribution, then the system is stable with probability one if
\begin{equation}\label{eq:mode dependent stability condition}
    \sum_{s} \overrightarrow{p_s}\ln{\mu_s}+p^\infty_s\ln(1-\alpha_s)< 0.
\end{equation}
\end{thm}
\begin{proof}
We start the proof similar to that of Theorem \ref{thm:main}. For a time horizon $k$, suppose there are $m$ mode jumps so far at time instants $k_i,i\in\{1,...,m\}$ such that $s_{k_i}\neq s_{k_i-1}$, i.e., $k_i$ is the time instant that the $i$-th jump just occurred. Then we  have
\begin{equation}
    V_{s_{k}}(x_k)\leq (1-\alpha_{s_{k_m}})^{k-k_{m}}V_{s_{k_{m}}}(x_{k_{m}}),
\end{equation}
from equation (\ref{eq:same mode decrease 1}) since the mode remains the same during the time interval $[k_{m},k]$, i.e. $s_t = s_k$ for all $t\in[k_m,k]$. Applying equations \eqref{eq:same mode decrease 1} and (\ref{eq: mode jump decrease 1}) recursively, we have
\begin{equation}\label{eq:mode dependent}
    \begin{split}
       &V_{s_{k}}(x_k)\leq  (1-\alpha_{s_{k_m}})^{k-k_{m}}\mu_{s_{k_m}} V_{s_{k_{m}-1}}(x_{k_{m}})\\
       &\leq(1-\alpha_{s_k})^{k-k_{m}}\mu_{s_k}\\
       & ~~~~~~~~~~~~~\times (1-\alpha_{s_{k_m-1}})^{k_m-k_{m-1}}\mu_{s_{k_m-1}}V_{s_{k_{m-1}-1}}(x_{k_{m-1}})
       \\&\leq \dots \leq V_{\hat{s}}(x_0)\prod_{i=0}^m (1-\alpha_{s_{k_i}})^{k_{i+1}-k_i}\mu_{s_{k_i}},
    \end{split}
\end{equation}
where $k_{m+1}=k,k_0=0$ and $\mu_{s_{k_0}}=1$. Among the $m$ number of jumps, we denote $m_s$ as the number of jumps to mode $s$ from a different mode. It is immediate that $\sum_s m_s=m$. Furthermore, we denote the time that the system spends in mode $s$ up to $k$ as $k_s$. By definition, $\sum_s k_s=k$. Then from \eqref{eq:mode dependent} we have
\begin{equation}\label{eq:mode dependent 1}
    \begin{split}
    V_{s_{k}}(x_k)&\leq V_{\hat{s}}(x_0)\prod_{i=0}^m (1-\alpha_{s_{k_i}})^{k_{i+1}-k_i}\mu_{s_{k_i}}\\
    &=V_{\hat{s}}(x_0)\prod_{s=1}^N \mu_s^{m_s}(1-\alpha_s)^{k_s}\\
    &=V_{\hat{s}}(x_0)(\prod_{s=1}^N \mu_s^{\frac{m_s}{k}}(1-\alpha_s)^{\frac{k_s}{k}})^k.
    \end{split}
\end{equation}
Similar to the proof of Theorem \ref{thm:main}, to prove the system is stable with probability one, we need to show that 
$$
P[\lim_{k\rightarrow\infty}(\prod_{s=1}^N \mu_s^{\frac{m_s}{k}}(1-\alpha_s)^{\frac{k_s}{k}})^k=0]=1,
$$
which implies that 
\begin{equation}\label{eq:mode dependent 2}
    \begin{split}
    &P[\lim_{k\rightarrow\infty}\prod_{s=1}^N \mu_s^{\frac{m_s}{k}}(1-\alpha_s)^{\frac{k_s}{k}}< 1]=1\\
        \Longleftrightarrow & P[\lim_{k\rightarrow\infty}{\sum_s \frac{m_s}{k}\ln(\mu_s)+\frac{k_s}{k}\ln(1-\alpha_s)}< 0]=1.
    \end{split}
\end{equation}
According to the law of large numbers \cite{bremaud2013markov}, we know that   
$$
P[\lim_{k\rightarrow\infty}\frac{m_s}{k}=\overrightarrow{p_s}]=1, \text{ and }
$$
$$
P[\lim_{k\rightarrow\infty}\frac{k_s}{k}=p^\infty_s]=1.
$$
Therefore, from \eqref{eq:mode dependent 2} and similar derivations that reach \eqref{eq:mode dependent thm 1}, we prove that if \eqref{eq:mode dependent stability condition} holds, the system is stable with probability one. 
\end{proof}
It is worth noticing that if $\mu_s=\mu$ and $\alpha_s=\alpha$ for any mode $s$, then condition \eqref{eq:mode dependent thm 1} becomes
\begin{equation}\label{eq:mode dependent 3}
       \sum_{s} \overrightarrow{p_s}\ln{\mu}+p^\infty_s\ln(1-\alpha)< 0.
\end{equation}
Since $\sum_s\overrightarrow{p_s}=P_{jump}$ and $\sum_sp^\infty_s=1$, equation \eqref{eq:mode dependent 3} reduces to 
$$
P_{jump}\ln{\mu}+\ln(1-\alpha)< 0.
$$
Therefore, we recover the stability condition from Theorem \ref{thm:main}.

\subsection{Computation of $\alpha$ and $\mu$}
The stability conditions in Theorem \ref{thm:main} and Theorem \ref{thm:main2} hinge  on the existence of multiple Lyapunov functions for each mode $s$ as well as   constants $\alpha$ and $\mu$ or $\alpha_s$ and $\mu_s$. This subsection gives an algorithm to find these Lyapunov functions and constants.

The Lyapunov function in mode $s$ can be formed as follows:  
	\begin{align}	
	V_{s}(x)=x'M_{s}x,
	\end{align}
where $M_{s}$ is a positive definite matrix.

We have 
	\begin{align}	
	V_{s}(x_{k+1})&=x_{k+1}'M_{s}x_{k+1}\nonumber\\&\le (1-\alpha_s)V_{s}(x_{k})=(1-\alpha_s)x_{k}'M_{s}x_{k}, 
	\end{align}
where $x_{k}:=x(k)$ for simplicity.
	
Therefore, $M_{s}$ and the largest $\alpha_{s}$ can be computed from the following bilinear optimization problem: 
	\begin{align}	
	\begin{split}
	& \max\limits_{\alpha_{s}, M_{s}} ~~~~\alpha_{s} \\
	& \textrm{s.t.} ~~~~~M_{s}'=M_{s}\succeq I,\\
	& ~~~~~~~~~A_{s}'M_{s}A_{s}\preceq (1-\alpha_{s})M_{s},
	\end{split}
	\label{LMI}
	\end{align}
where $I$ is the identity matrix.

When the state jump from mode $s$ to mode $s'$ ($s'\neq s$), we have
	\begin{align}	
	V_{s}(x_{k})=x_{k}'M_{s}x_{k}\le \mu_{s, s'} V_{s'}(x_{s})=\mu_{s, s'} x_{k}'M_{s'}x_{k}. 
	\end{align}
After $M_{s}$ is computed for every mode $s$ using (\ref{LMI}), the smallest $\mu_{s, s'}$ can be computed from the following SDP problem:
	\begin{align}	
	\begin{split}
	& \min\limits_{\mu_{s, s'}} ~~~~\mu_{s, s'} \\
	& \textrm{s.t.} ~~~~~\mu_{s, s'}>1\\
	&~~~~~~~~~M_{s}\preceq \mu_{s, s'} M_{s'}.
	\end{split}
	\label{LMI2}
	\end{align}

\subsection{Policy Synthesis for Stability With Probability One}
This subsection introduces computation approach to find a policy such that an MDP-JLS is guaranteed to be stable with probability one, based on stability conditions in Theorem \ref{thm:main} and Theorem \ref{thm:main2}. In addition to stability, we assign a cost function $c:S\rightarrow\mathbb{R}$ to each mode, such that the stabilizing policy should also minimize a  cost $\sum_s c(s) p_s^\infty $, which represents an average cost over states when the underlying DTMC reaches its stationary distribution. We start with policy synthesis following stability conditions stated in Theorem \ref{thm:main}. 

\begin{thm}\label{thm:modeIndLp}
In case of perfect model knowledge of the MDP, given stability constants $0<\alpha<1$, $\mu>1$ and given small $\epsilon > 0$, Problem~\ref{problem:stability with probability one} can be formulated as the linear optimization problem

\begin{align}\label{eq:plpformulation}
& \underset{\hat{\pi},p^\infty,\hat{P}}{\text{minimize}}
& & \sum_s c(s) p_s^\infty \\
& \text{subject to }
& &  \hat{\pi} \mathbf{1} = p^\infty, \quad \mathbf{1}' p^\infty = 1,  \label{stochast_constr} \\
& & & \hat{P}' \mathbf{1} = p^\infty, \label{stat_dist_constr}\\
& & & \hat{P}_{ij} = \sum_\sigma T(i , \sigma , j) \hat{\pi}(i,\sigma) \quad \forall i,j \in S,\label{mdp_impl_constr}\\
& & & 1 - Tr ( \Hat{P}) < \frac{\ln(\frac{1}{1-\alpha})}{\ln(\mu)}, \label{prob_one_constr} \\
& & & \hat{\pi} \geq 0 , \quad p^\infty \geq \epsilon \label{pos_constr_pi_p}, 
\end{align}
where $\hat{P} \in \mathbb{R}^{|S| \times |S|}$, $p^\infty \in \mathbb{R}^{|S|}$ and $\hat{\pi} \in \mathbb{R}^{|S| \times |\Sigma|}$. If we denote by $\hat{\pi}_{opt},p_{opt}^\infty,\hat{P}_{opt}$ the optimal solutions, then the policy $\pi$, the induced Markov matrix $P$ and the stationary distribution $p^\infty$ solutions of problem~\ref{problem:stability with probability one} are given by 
\[ \pi = diag(p_{opt}^\infty)^{-1} \hat{\pi}_{opt}, \]
\[ P = diag(p_{opt}^\infty)^{-1} \hat{P}_{opt}, \text{ and }\]
\[ p^\infty = p_{opt}^\infty. \]
\end{thm}


\begin{proof}
This LP formulation comes immediately from the change of variables
\begin{align*}
    \hat{P} &= diag(p^\infty) P, \text{ and } \\
    \hat{\pi} &= diag(p^\infty) \pi.
\end{align*}
Observe that $\forall i,j \in S, \: \hat{P}_{ij} = p_i^\infty P_{ij}$ and $ \forall \sigma \in \Sigma, \: \hat{\pi}_{ik} = p_i^\infty \pi_{i\sigma}$.

By definition $p^\infty > 0$ and $\pi \geq 0$ is equivalent to $\hat{\pi} \geq 0$ and $\quad p^\infty \geq \epsilon$ for a fixed small $\epsilon > 0$. This gives constraint~\eqref{pos_constr_pi_p}.

$\pi \mathbf{1} = 1 \Longleftrightarrow \hat{\pi} \mathbf{1} = p^\infty$ by left multiplication with invertible matrix $diag(p^\infty)$. Since $p^\infty$ is a probability distribution,  $\mathbf{1}' p^\infty = 1$. This proves constraint~\eqref{stochast_constr}.

By definition of the stationary distribution $p^\infty$,  we have
\[ 
    \sum_i p_i^\infty P_{ij} = \sum_i \hat{P}_{ij} = p_j^\infty \text{ for all } j \in S.
\]
Therefore $(\pi^\infty)' P = (\pi^\infty)' \Longleftrightarrow \mathbf{1}' \hat{P} = (\pi^\infty)'$. This gives constraint~\eqref{stat_dist_constr}.

Since $p^\infty > 0$, the constraint~\eqref{mdp_impl_constr} is given by

\begin{equation*}
\begin{aligned}
 &P_{ij} = \sum_\sigma T(i,\sigma,j) \pi(i,\sigma) \\
   \Longleftrightarrow &p_i^\infty P_{ij} = \sum_\sigma T(i,\sigma,j) (p_i^\infty \pi(i,\sigma)) \\
   \Longleftrightarrow &\hat{P}_{ij} = \sum_\sigma T(i,\sigma,j) \hat{\pi}(i,\sigma),
\end{aligned}
\end{equation*}
for all $i,j \in S$. 

Finally, constraint~\eqref{prob_one_constr} is given by Theorem \ref{thm:main} 
\[
P_{jump} = 1 - \sum_ip_i^\infty P_{ii} = 1 - Tr ( \Hat{P}) < \frac{\ln(\frac{1}{1-\alpha})}{\ln(\mu)}.
\]
\end{proof}

Given the mode-dependent stability coefficients $\alpha_s$ and $\mu_s$, Theorem \ref{thm:main2} can be formulated as another LP where a feasible solution contains a policy that guarantees stability of the system with probability one.

\begin{thm}\label{thm:modeDepLp}
Consider an MDP-JLS  whose mode $s\in S$ makes transitions following a MDP $\mathcal{M}=(S,\hat{s},\Sigma,T)$ and the dynamics in each modes as in \eqref{eq:switched systems} with given mode-dependent constants $0<\alpha_s<1$ and $\mu_s>1$ for all $s\in S$. Given small $\epsilon > 0$,  Problem~\ref{problem:stability with probability one} can be formulated as the linear optimization problem

\begin{align} \label{eq:mlpformulation}
& \underset{\hat{\pi},\pi^\infty,\hat{P}}{\text{minimize}}
& & \sum_s c(s) p_s^\infty \\
& \text{s.t }
& &  \hat{\pi} \mathbf{1} = p^\infty, \quad \mathbf{1}' p^\infty = 1, \label{eq:mlpformulation:1} \\
& & & \hat{P}' \mathbf{1} = p^\infty, \\
& & & \hat{P}_{ij} = \sum_\sigma T(i , \sigma , j) \hat{\pi}(i,\sigma) \quad \forall i,j \in S,\\
& & & \sum_i (\sum_{j \neq i} \hat{P}_{ji}) ln(1-\mu_i) + p_i^\infty ln(1 -\alpha_i) < 0, \label{mode_dep_prob_one_constr}\\
& & & \hat{\pi} \geq 0 , \quad p^\infty \geq \epsilon \label{eq:mlpformulation:2} , 
\end{align}
where $c : S \rightarrow \mathbb{R}$ is an optimization criteria over the set of solutions, $\hat{P} \in \mathbb{R}^{|S| \times |S|}$, $p^\infty \in \mathbb{R}^{|S|}$ and $\hat{\pi} \in \mathbb{R}^{|S| \times |\Sigma|}$. If $\hat{\pi}_{opt},p_{opt}^\infty,\hat{P}_{opt}$ are optimal solutions, then the policy $\pi$, the induced Markov matrix $P$ and the stationary distribution $p^\infty$ solutions of problem~\ref{problem:stability with probability one} are given by 
\[ \pi = diag(p_{opt}^\infty)^{-1} \hat{\pi}_{opt}, \]
\[ P = diag(p_{opt}^\infty)^{-1} \hat{P}_{opt}, \text{ and } \]
\[ p^\infty = p_{opt}^\infty. \]
\end{thm}

\begin{proof}
This is an immediate consequence of Theorem \ref{thm:main2}. Observe that $\forall i,j \in S, \: \hat{P}_{ij} = p_i^\infty P_{ij}$. Therefore,
$\overrightarrow{p_s} =  \sum_{s'\neq s}P(s',s)p^\infty_{s'} = \sum_{s'\neq s} \hat{P}(s',s)$. The constraint~\eqref{mode_dep_prob_one_constr} is immediately obtained from the last observation and Theorem \ref{thm:main2}. 

The others inequalities constraints can be derived using the same proof sketch as in Theorem \ref{thm:modeIndLp}.

\end{proof}

If the state-based cost is not of interest, then a feasibility program with the constraints from \eqref{stochast_constr} to  \eqref{pos_constr_pi_p} or from  \eqref{eq:mlpformulation:1} to  \eqref{eq:mlpformulation:2} can be solved instead.

\subsection{Stability Guarantee With Imperfect Model Knowledge}
In many cases, the true MDP $\mathcal{M}$ that governs the switch dynamics may not be precisely known but can only be estimated through statistical experiments. As a result, it may only be possible to obtain an approximated model $\bar{\mathcal{M}}$ such that the transition probabilities in  ${\mathcal{M}}$ is known to lie in some neighborhood of that in $\bar{\mathcal{M}}$. To make such notion of approximation more precise, we first define $\Delta-$approximation in MDPs \cite{brafman2002r}.
\begin{definition}\label{dfn:alphaMDP}
Let $\mathcal{M}=(S,\hat{s},\Sigma,T)$ and $\bar{\mathcal{M}}=(\bar{S},\bar{\hat{s}},\bar{\Sigma},\bar{T})$ be two MDPs. $\bar{\mathcal{M}}$ is a $\Delta-$approximation of $\mathcal{M}$ if
\begin{itemize}
	\item $S=\bar{S}$, $\hat{s}=\bar{\hat{s}}$, $\Sigma=\bar{\Sigma}$. That is, they share the same state space, initial condition and action space;
	\item $|T(s,\sigma,s')-\bar{T}(s,\sigma,s')|\leq\Delta$ for any $s,s'$ and $\sigma$;
	\item $T(s,\sigma,s')>0$ if and only if $\bar{T}(s,\sigma,s')>0$ for any $s,s'$ and $\sigma$.
\end{itemize}
\end{definition}
By definition, it is not hard to see that if $\bar{\mathcal{M}}$ is a $\Delta-$approximation of $\mathcal{M}$, then with the same policy $\pi$, we have that $\mathcal{C}$, the induced Markov chain from  $\mathcal{M}$ is a $\Delta-$approximation of  $\bar{\mathcal{C}}$, the induced Markov chain from $\bar{\mathcal{M}}$.

When the transition probabilities of two DTMCs are close to each other, their stationary distribution is also close as shown the following theorem. 
\begin{thm}\cite{funderlic1986sensitivity}\label{thm:stationary distribution}
Let $\bar{\mathcal{C}}$ and $\mathcal{C}$ be two DTMCs both with $N$ states, and transition matrices are $\bar{P}$ and $P=\bar{P}-F$, respectively.  Then the stationary distribution $p^\infty$ and $\bar{p}^\infty$ satisfy
\begin{equation}\label{equation:stationary distribution bound}
    |p_i^\infty-\bar{p}_i^\infty|\leq ||F||_\infty\max_{j}|h_{ij}^\#|,\text{ for each } i\in\{1,...,N\}
\end{equation}
where $N=|S|,p_i^\infty:=p(s_i)^\infty$, $h_{ij}^\#$ is the $(i,j)$ entry of a matrix $H^\#$ which is the group inverse of $H=I-\bar{P}$.
\end{thm}
The group inverse $H^\#$ of a matrix $H$, if it exists, can be uniquely determined by three equations $HH^\#H=H,H^\#HH^\#=H^\#$, and $HH^\#=H^\#H$ \cite{funderlic1986sensitivity}. For a transition matrix $P$, $H^\#$ always exists and can be computed as the following \cite{meyer1975role}. We can write $H$ as
$$
H = \begin{bmatrix} 
U & c \\
d' & \alpha 
\end{bmatrix},
$$
where $U\in\mathbb{R}^{(N-1)\times(N-1)}$, and
$$
h' = d'U^{-1}, \delta = -h'U^{-1}1_{N-1},\beta = 1-h'1_{N-1},G=U^{-1}-\frac{\delta}{\beta}I.
$$
Here, $\delta$ and $\beta$ are nonzero scalars. Then $H^\#$ can be calculated by
\begin{equation}\label{eq:H inverse}
H^{\#} = \begin{bmatrix}
U^{-1}+\frac{U^{-1}1_{N-1}h'U^{-1}}{\delta}-\frac{G1_{N-1}h'G}{\delta} & -\frac{G1_{N-1}}{\beta}\\
\frac{h'G}{\beta} & \frac{\delta}{\beta^2}
\end{bmatrix},
\end{equation}
The following lemma is a direct application of Theorem \ref{thm:stationary distribution} and the definition of $\Delta-$approximation.
\begin{lemma}\label{lemma:stationary distribution bound}
Given two DTMCs $\mathcal{C}$ and $\bar{\mathcal{C}}$ with $N$ states, transition matrices $P$ and $\bar{P}$, respectively. If $\bar{\mathcal{C}}$ is a $\Delta-$approximation of $\mathcal{C}$, then their stationary distribution $p^\infty$ and $\bar{p}^\infty$ satisfy
\begin{equation}
 |p_i^\infty-\bar{p}_i^\infty|\leq N\Delta\max_j|h_{ij}^\#|,   
\end{equation}
\end{lemma}
Suppose we have a $\Delta-$approximation $\bar{\mathcal{M}}$ of the underlying MDP ${\mathcal{M}}$. Since it is not possible to know ${\mathcal{M}}$  precisely, we can only find some policy $\pi$ based on the estimated model $\bar{\mathcal{M}}$ and apply it to the true MDP $\mathcal{M}$. As illustrated in Theorem \ref{thm:main}, to be able to guarantee the stability of the system, we have to bound the difference between $P_{jump}$ and $\bar{P}_{jump}$, which are probabilities of mode jumps in the ${\mathcal{M}}$ and $\bar{\mathcal{M}}$, respectively.

Recall that 
$
P_{jump} = 1 - \sum_ip_i^\infty P_{ii}\text{ and } \bar{P}_{jump} = 1 - \sum_i\bar{p}^\infty_i\bar{P}_{ii}.
$
Let $P_{ii}=\bar{P}_{ii}+\delta_i$ and  $p_i^\infty=\bar{p}_i^\infty+\epsilon_i$. Once we have the estimated model $\bar{\mathcal{M}}$ which is a $\Delta-$approximation of $\mathcal{M}$, given any policy $\pi$, it is possible to obtain the transition matrix $\bar{P}$ and the stationary distribution $\bar{p}^\infty$. Furthermore, by the definition of $\Delta-$approximation and Lemma \ref{lemma:stationary distribution bound}, we know that 
\begin{equation}
    \begin{split}
        -\Delta\leq&\delta_i\leq\Delta, \text{ and}\\
        -N\Delta\max_j|h_{ij}^\#|\leq&\epsilon_i\leq N\Delta\max_j|h_{ij}^\#|.
    \end{split}
\end{equation}

Then we  have 
\begin{equation}\label{eq:P_jump}
\begin{split}
&{P}_{jump}-\bar{P}_{jump} = \sum_ip_i^\infty P_{ii}-\bar{p}_i^\infty\bar{P}_{ii}\\
                        &= \sum_i(\bar{p}_i^\infty+\epsilon_i)(\bar{P}_{ii}+\delta_i)-\bar{p}_i^\infty \bar{P}_{ii}\\
                        &= \sum_i \bar{p}_i^\infty\delta_i+\bar{P}_{ii}\epsilon_i+\epsilon_i\delta_i\\
                        &\leq \sum_i \bar{p}_i^\infty\Delta+\bar{P}_{ii}N\Delta\max_j|h_{ij}^\#|+N\Delta^2\max_j|h_{ij}^\#|.
\end{split}
\end{equation}
Therefore, we have the following theorem to guarantee the policy we find from the estimated model $\bar{\mathcal{M}}$ is able to stabilize the switched system whose switching is governed by the true model $\mathcal{M}$.
\begin{thm}\label{thm:approximation}
Let $\mathcal{M}=(S,\hat{s},\Sigma,T)$ and $\bar{\mathcal{M}}=(S,{\hat{s}},{\Sigma},\bar{T})$ be two MDPs where $\bar{\mathcal{M}}$ is a $\Delta-$approximation of $\mathcal{M}$. A policy $\pi$ incurs two DTMCs $\mathcal{C}=(S,\hat{s},P)$ and $\bar{\mathcal{C}}=(S,\hat{s},\bar{P})$, respectively. Then the system is stable with probability one with the policy $\pi$ if conditions \eqref{eq:same mode decrease} and \eqref{eq: mode jump decrease} are satisfied and
\begin{equation}
\begin{split}
   \bar{P}_{jump}+ \sum_i \bar{p}_i^\infty\Delta+\bar{P}_{ii}N\Delta\max_j|h_{ij}^\#|+& N\Delta^2\max_j|h_{ij}^\#|\\
   &~~~~~~~< \frac{\ln(\frac{1}{1-\alpha})}{\ln(\mu)}.     
\end{split}
\end{equation}
\end{thm}
\begin{proof}
This is a direct result of Theorem \ref{thm:main} and equation \eqref{eq:P_jump}.
\end{proof}
As a result, to synthesize a policy $\pi$ from the MDP $\bar{\mathcal{M}}$, the requirement for $\bar{P}_{jump}$ is changed to 
\begin{equation}
\begin{split}
   \bar{P}_{jump}&< \frac{\ln(\frac{1}{1-\alpha})}{\ln(\mu)}-\sum_i \bar{p}_i^\infty\Delta\\&~~~~+\bar{P}_{ii}N\Delta\max_j|h_{ij}^\#|+N\Delta^2\max_j|h_{ij}^\#|.   
\end{split}
\end{equation}

Similar theorem can also be derived for mode-dependent stability conditions in \eqref{eq:same mode decrease 1},\eqref{eq: mode jump decrease 1} and \eqref{eq:mode dependent stability condition}. First we denote $\overrightarrow{\bar{p_i}}$ as the probability to jump to $s_i$ for the estimated MDP $\mathcal{M}$ with some policy $\pi$, we observe from \eqref{eq:jump to state s} that
\begin{equation}\label{eq:P_arrow}
\begin{split}
\overrightarrow{p_i}-\overrightarrow{\bar{p_i}}&=\sum_{j\neq i}P_{ji}p^\infty_{j}-\bar{P}_{ji}\bar{p}^\infty_{j},\\
&\leq\sum_{j\neq i}(\bar{P}_{ji}+\Delta)(\bar{p}^\infty_{j}+N\Delta\max_j|h_{ij}^\#|)-\bar{P}_{ji}\bar{p}^\infty_{j},\\&=\sum_{j\neq i}\bar{P}_{ji}N\Delta\max_j|h_{ij}^\#|+\bar{P}_{ji}\Delta+N\Delta^2\max_j|h_{ij}^\#|.
\end{split}
\end{equation}
We are now ready to give the following theorem similar to Theorem \ref{thm:approximation} for the mode-dependent parameters case.
\begin{thm}\label{thm:approximation 1}
Let $\mathcal{M}=(S,\hat{s},\Sigma,T)$ and $\bar{\mathcal{M}}=(S,{\hat{s}},{\Sigma},\bar{T})$ be two MDPs where $\bar{\mathcal{M}}$ is a $\Delta-$approximation of $\mathcal{M}$. A policy $\pi$ incurs two DTMCs $\mathcal{C}=(S,\hat{s},P)$ and $\bar{\mathcal{C}}=(S,\hat{s},\bar{P})$, respectively. Consider an MDP-JLS  whose mode $s\in S$ makes transitions following the MDP $\mathcal{M}=(S,\hat{s},\Sigma,T)$ and the dynamics in each modes as in \eqref{eq:switched systems}. The MDP-JLS is stable with probability one with the policy $\pi$ if conditions \eqref{eq:same mode decrease 1} and \eqref{eq: mode jump decrease 1} are satisfied and
\begin{equation}\label{eq:mode dependent stability condition 1}
\begin{split}
    &\sum_{i}[(\overrightarrow{\bar{p_i}}+\sum_{j\neq i}\bar{P}_{ji}N\Delta\max_j|h_{ij}^\#|+\bar{P}_{ji}\Delta+N\Delta^2\max_j|h_{ij}^\#|)\\
    &~~~~~~~~~~~~~~\ln{\mu_i}+(\bar{p}^\infty_i+N\Delta^2\max_j|h_{ij}^\#|)\ln(1-\alpha_i)]< 0.    
\end{split}
\end{equation}
\end{thm}
\begin{proof}
This is a direct result of \eqref{eq:mode dependent stability condition} and \eqref{eq:P_arrow}.
\end{proof}

The counterpart of policy synthesis formulations for uncertain MDPs can be derived in a way similar to Theorem \ref{thm:modeIndLp} and Theorem \ref{thm:modeDepLp} but with conditions stated in Theorem \ref{thm:approximation} and Theorem \ref{thm:approximation 1}, respectively.

\section{\added{Numerical Examples}}\label{sec:Examples}


We demonstrate the proposed approach on two examples: vehicle formation and transportation networks. The simulations were performed on a computer with an Intel Core i9-9900K 3.60 GHz x 16 processors and 62.7 GB of RAM with MOSEK~\cite{mosek} as the SDP solver, GUROBI~\cite{gurobi} as the LP solver and using the CVX~\cite{cvx} interface.  

In each subsection, we show and compare the results of proposed methods with both mean-square stability and stability with probability one. For mean-square stability, we use both CD and SDP approaches. For the CD method, we initialize $V^0=I$ and $\pi^0$ to be uniform over all actions. CD methods could converge to a saddle point, and we add additional random term to each action uniformly selected over the interval $[-\delta, \delta]$, where $\delta>0$ is a small constant, to ensure that the procedure does not converge to a saddle point.

For  stability with probability one, we apply two linear programming-based methods for mode-independent and mode-dependent coefficients. The mode-independent LP is the  LP in (\ref{eq:plpformulation}) with mode-independent stability coefficients from Theorem \ref{thm:modeIndLp}. The mode-dependent LP is the LP in (\ref{eq:mlpformulation}) with mode-dependent stability coefficients from Theorem \ref{thm:modeDepLp}.  For the LP-based methods, $\epsilon$ is taken sufficiently small in order for the positive constraints to be satisfied. The mode-dependent/independent coefficients $\alpha_s$ and $\mu_s$ are generated with   (\ref{LMI}) and  (\ref{LMI2}).

\subsection{Vehicle Formation Example}
The example used in this section is adapted from \cite{arantzerExample}. We consider the vehicle formation example where the continuous-time dynamic is given by
\begin{equation}\label{vehicledynamic}
    \begin{cases}
        \dot{x}_1 = -x_1 + l_{13} (x_3 - x_1) \\
        \dot{x}_2 = l_{21}(x_1 - x_2) + l_{23} (x_3 - x_2) \\
        \dot{x}_3 = l_{32}(x_2 - x_3) + l_{34} (x_4 - x_3) \\
        \dot{x}_4 = -4x_4 + l_{43} (x_3 - x_4). \\
    \end{cases}
\end{equation}

The state $x_i$ represents the position of vehicle $i$ and the parameters $l_{ij}$ represent position adjustments based on distance measurements between the vehicles. We consider the discrete-time version of the model (\ref{vehicledynamic}) with sampling time $dt = 0.1$.

The team of vehicle can be modeled as an MDP-JLS with three modes corresponding to three different position adjustment parameters. There are two actions that trigger transitions between modes probabilistically. We are interested in the  MDP-JLS with the following characteristics.

\begin{center}
    
\begin{tabu} { | X[0.2,c] | X[0.85,l] | X[l] | }
 \hline
 Mode & Dynamics & $\alpha , \mu $\\
 \hline
 1  & $l_{13} = l_{32}=l_{34}=3$  $l_{21} = l_{23} = 5$ ~~~ $l_{43}=2$  &  $\alpha_1 = 0.21875$ $\mu_1 = 1.682$  \\
 \hline
 2  & $l_{13} = l_{43} = 0$  ~~~~$l_{21} = 0.5$ ~~~~~~~~~~~~~~~~~~~~~~$l_{23} = l_{32} =l_{34} = 0.5$   &  $\alpha_2 = 0.09375$ $\mu_2 = 1.885$ \\
 \hline
 3  & $l_{13} = l_{43} = 1$  $l_{21} = l_{34} = 3$  $l_{23} = l_{32} = 5$  &  $\alpha_3 = 0.21093$ $\mu_3 = 1.928$ \\
 \hline
\end{tabu} 

\end{center}
The parameters $\alpha , \mu$ for every modes have been found by solving the SDP in (\ref{LMI}) and (\ref{LMI2}) for each modes. Moreover, we consider that the transitions between modes are governed by an MDP with two actions and the transitions probabilities given by
\[
 T_1 = 
\begin{bmatrix}
    0.8       &     0.15    &   0.05  \\
    0.03      &     0.95     &   0.02  \\
    0.85      &     0.05     &   0.1  
\end{bmatrix} 
\text{and } T_2 = 
\begin{bmatrix}
    0.3       &     0.6    &   0.1  \\
    0.9      &     0.05     &   0.05  \\
    0.08      &     0.02     &   0.9  
\end{bmatrix}.
\]
The goal is to generate a policy that guarantees the system stability with probability one and in the mean-square sense. Using the SDP, CD, the mode-independent LP, and the mode-dependent LP methods, we synthesize such policies and summarize the results in Fig.~\ref{fig:numstate}. 
The evolution of state $x(t)$   in Fig.~\ref{fig:numstate} is of logarithmic scale to show how fast the policies generated by different methods converge.

\input{stateEvolution.tex}

\subsection{Transportation Example}
Consider the linear transportation example \cite{wu2019switched} \cite{rantzer2012optimizing} given by the continuous time dynamic $\dot{x} =  A x$ where
\begin{equation} \label{transportation_dyn}
    A = \begin{bmatrix}
        -1-l_{31} & l_{12} & 0 & 0 \\
        0         & 2-l_{12}-l_{32} & l_{23} & 0 \\
        l_{31}    & l_{32}          & 3-l_{23}-l_{43} & l_{34} \\
        0         &  0    & l_{43} & -4-l_{34}
    \end{bmatrix}.
\end{equation}
This model describes a transportation network connecting four buffers. The state $x$ represents the quantities of contents in the buffers and the parameter $l_{ij}$ determines the rate of transfer from buffer $j$ to buffer $i$. 

We consider a discrete-time version of (\ref{transportation_dyn}) with sample time $dt=0.1$. In particular,  there are two actions that may affect the rate of transfer probabilistically which result in four different
matrices $A$. This can be modeled as an MDP-JLS with transitions governed by an MDP with four discrete modes and two actions.   
The transition probabilities induced by the two actions are given by
\[
 T_1 = 
\begin{bmatrix}
    0.1 & 0.7 & 0.1 & 0.1 \\
    0.1 & 0.8 & 0.05 & 0.05 \\
    0.2 & 0.6 & 0.1 & 0.1 \\
    0.1 & 0.05 & 0.05 & 0.8
\end{bmatrix}
\]
and 
\[
 T_2 = 
\begin{bmatrix}
    0.8 & 0.05 & 0.05 & 0.1 \\
    0.3 & 0.15 & 0.4 & 0.15 \\
    0.1 & 0.1 & 0.7 & 0.1 \\
    0.1 & 0.7 & 0.1 & 0.1 
\end{bmatrix}.
\]
We run the policy synthesis using the four different methods for 25 different switched linear systems with the same MDP as defined above but different continuous dynamics. The feasibility of finding a policy stabilizing each system is evaluated by the number of times the method can find a stabilizing policy and the average run time.


{
\centering
\captionof{table}{Linear transportation network results with 25 different systems. Each mode in each system has a spectral radius between $[0.63,0.98[$ with a mean spectral radius of 0.8.} \label{tab:lowspectralexp} 
\vskip 0.1in
\begin{tabular}{ccc}
  \hline
   &  Successful cases &  Mean Time\\
  \hline
  CD &  25    & 1.32s \\
  \hline
  SDP  & 23     &  0.044s\\
  \hline
  Mode-independent LP & 2  & 0.027s \\
  \hline
  Mode-dependent LP & 22  & 0.029s \\
  \hline
\end{tabular}
\vskip 0.1in
}
For mean-square stability, Table \ref{tab:lowspectralexp}  shows that the CD method always manages to find a policy that stabilizes the system in all cases given sufficient amount of time. The SDP approximation is not as reliable as CD, but can find a feasible policy with better performances in term of computational time. The LP-based methods may fail to find a stabilizing policy. When they are able to find such a policy, the computational time for the task is significantly less than that of the CD-based method. The insight is that, the number of the variables of the approach based on LP only depends on the number of modes, but the number of the variables of the approach based on SDP depends on both the number of modes and the dimension of the continuous system.

For example, consider the linear transportation example (\ref{transportation_dyn}) where instead of having just four buffers and four modes, we have 12 buffers and 8 modes. Running the CD method and the mode-dependent LP method on multiple instances of this example gives a mean execution time of 2.3s for the CD method and 0.17s for the LP method. If we add four more buffers (16 states per mode) with the same number of modes and actions, we obtain in average 0.17s for the mode-dependent LP to find a feasible policy and 7.9s for the CD method. We perform the same experiment with 20 states per mode and obtain 0.17s for the LP mean execution time and 19.1s for the CD method mean execution time. Finally, using the same settings as previously described but with 40 states per mode, the mode-dependent LP is able to find a feasible policy with an average computational time of 0.23s when the CD method has an average of 1136.8s. Table \ref{tab:lp_cd_state_evol} summarizes the computation time results.
{
\captionof{table}{Evolution of mean computational time for the CD and LP (\ref{eq:mlpformulation}) methods with the number of states per mode. } \label{tab:lp_cd_state_evol} 
\centering
\vskip 0.1in
\begin{tabular}{ccc}
  \hline
   Number of states &  Mode-dependent LP &  CD\\
  \hline
  12 &  0.17s    & 2.3s \\
  \hline
  16  & 0.17s     &  7.9s\\
  \hline
  20  & 0.17s  & 19.1s \\
  \hline
  40  & 0.23s  & 1136.8s \\
  \hline
\end{tabular}

\vskip 0.1in
}
Based on Table \ref{tab:lp_cd_state_evol}, to find a policy to guarantee stability with probability one with coefficients $\alpha_s$ and $\mu_s$, one could first try to solve the LP-based methods as they are much more faster than the other methods. If the LP methods fail to find a stabilizing policy, then one can try the CD-based method which is more reliable but could have much longer run time.

\section{Conclusion}\label{sec:Conclusion}
In this paper, we consider \added{a class of switched linear systems} whose mode switches are governed by a Markov decision process (MDP) and  we name such systems MDP-JLS for brevity. The objective is to find a policy in the MDP to stabilize the MDP-JLS. Given a policy, an MDP reduces to a discrete-time Markov chain, and an MDP-JLS becomes a Markov jump linear system (MJLS). For mean-square stability, we leverage the existing stability conditions in MJLSs and propose semidefinite programming (SDP)-based approaches to compute the stabilizing policy. \added{For stability with probability one, we derive new sufficient stability conditions based on which we formulate linear programs to find the stabilizing policy. We also extend the policy synthesis results to MDP-JLSs with uncertain transition probabilities and the optimization of the  expected   state-dependent cost.  The numerical experiments validate the proposed approaches.} 
 
\added{
This paper opens the door to study a class of switched systems whose switches are governed by MDPs. For future work, we will continue to investigate how to incorporate additional temporal logic constraint on mode switches. We will also study policy synthesis with partially observable modes.}

\bibliographystyle{IEEEtran}
\bibliography{ref}
\begin{IEEEbiography}[{\includegraphics[width=1in,height=1.25in,clip,keepaspectratio]{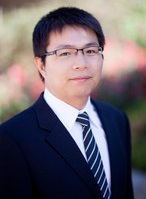}}]{Bo Wu}
received a B.S. degree from Harbin Institute of Technology, China, in 2008, an M.S. degree from Lund University, Sweden, in 2011 and a Ph.D. degree from the University of Notre Dame, USA, in 2018, all in electrical engineering. He is currently a postdoctoral researcher at the Oden Institute for Computational Engineering and Sciences at the University of Texas at Austin. His research interest is to apply formal methods, learning, and control in autonomous systems, such as robotic systems, communication systems, and human-in-the-loop systems, to provide privacy, security, and performance guarantees.
\end{IEEEbiography}

\begin{IEEEbiography}[{\includegraphics[width=1in,height=1.25in,clip,keepaspectratio]{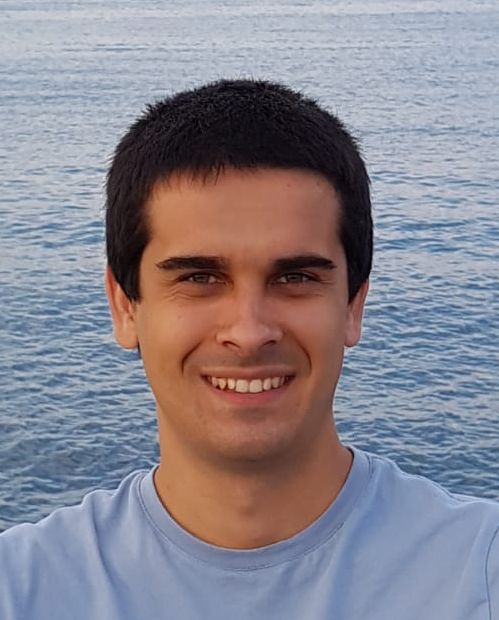}}]{Murat Cubuktepe} joined the Department of Aerospace Engineering at the University of Texas at Austin as a Ph.D. student in Fall 2015. He received his B.S degree in Mechanical Engineering from Bogazici University in 2015 and his M.S degree in Aerospace Engineering and Engineering Mechanics from the University of Texas at Austin in 2017. His current research interests are verification and synthesis of parametric and partially observable probabilistic systems. He also focuses on applications of convex optimization in formal methods and controls.
\end{IEEEbiography}
\begin{IEEEbiography}[{\includegraphics[width=1in,height=1.25in,clip,keepaspectratio]{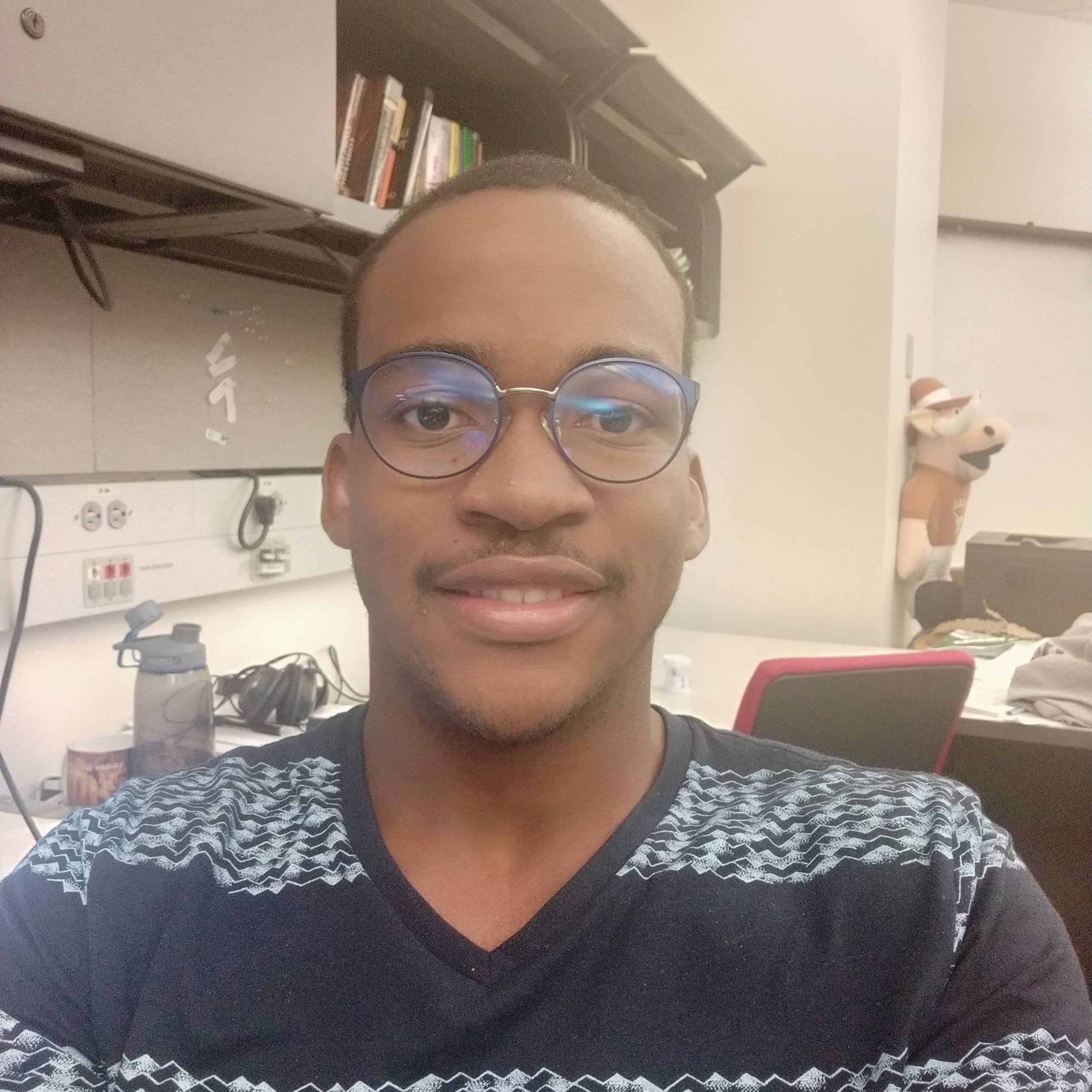}}]{Franck Djeumou} joined the Department of Electrical and Computer Engineering at the University of Texas at Austin as a Ph.D. student in Fall 2018. He received his B.S and M.S degrees in Aerospace Engineering from ISAE-SUPAERO, France, in 2018. He also received a M.S degree in Computer Science from \'Ecole polytechnique, France, in 2017. His current research interests include formal methods, autonomous systems and control systems.
\end{IEEEbiography}
\begin{IEEEbiography}[{\includegraphics[width=1in,height=1.25in,clip,keepaspectratio]{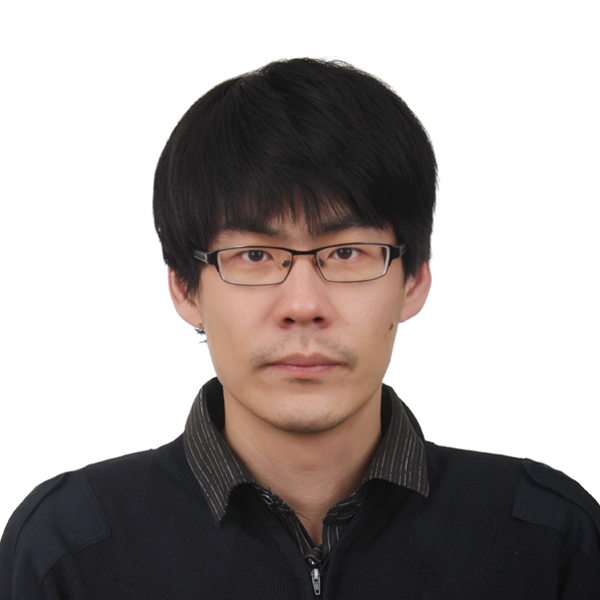}}]{Zhe Xu}
received the B.S. and M.S. degrees in Electrical Engineering from Tianjin University, Tianjin, China, in 2011 and 2014, respectively. He received the Ph.D. degree in Electrical Engineering at Rensselaer Polytechnic Institute, Troy, NY, in 2018. He is currently a postdoctoral researcher in the Oden Institute for Computational Engineering and Sciences at the University of Texas at Austin, Austin, TX. His research interests include formal methods, autonomous systems, control systems and reinforcement learning. 
\end{IEEEbiography}
\begin{IEEEbiography}[{\includegraphics[width=1in,height=1.25in,clip,keepaspectratio]{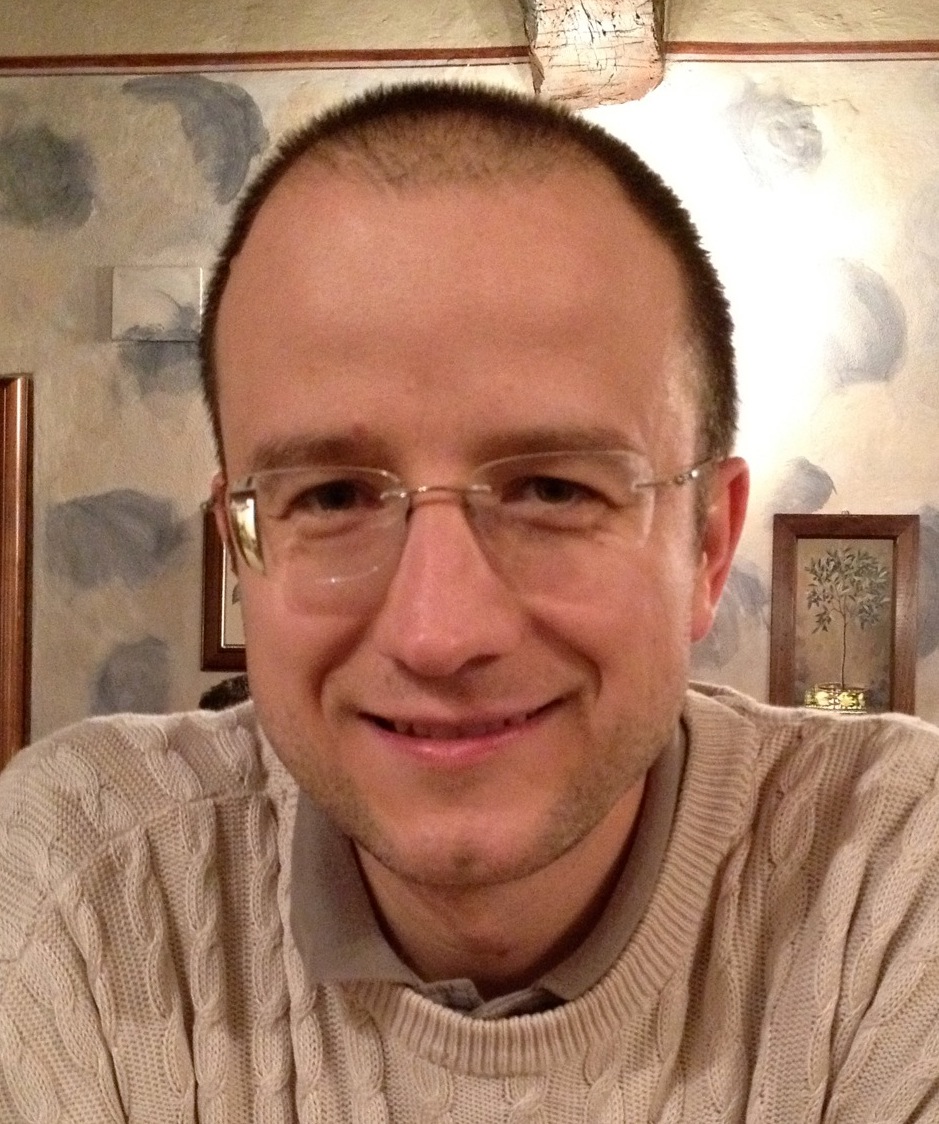}}]{Ufuk Topcu}
Ufuk Topcu joined the Department of Aerospace Engineering at the University of Texas at Austin as an assistant professor in Fall 2015. He received his Ph.D. degree from the University of California at Berkeley in 2008. He held research positions at the University of Pennsylvania and California Institute of Technology. His research focuses on the theoretical, algorithmic and computational aspects of design and verification of autonomous systems through novel connections between formal methods, learning theory and controls. 
\end{IEEEbiography}
\end{document}

%% file: stateEvolution.tex
\begin{figure}
    \centering
    
\begin{tikzpicture}[scale=0.9]

\definecolor{color1}{rgb}{1,0.498039215686275,0.0549019607843137}
\definecolor{color2}{rgb}{0.172549019607843,0.627450980392157,0.172549019607843}
\definecolor{color0}{rgb}{0.12156862745098,0.466666666666667,0.705882352941177}
\definecolor{color3}{rgb}{0.83921568627451,0.152941176470588,0.156862745098039}

\begin{groupplot}[group style={group size=2 by 2,vertical sep=2cm, horizontal sep=1.5cm},width=5cm,height=4cm]
\nextgroupplot[
log basis y={10},
tick align=outside,
tick pos=both,
title={SDP},
x grid style={lightgray!92.02614379084967!black},
xmajorgrids,
xmin=-2.45, xmax=51.45,
xtick style={color=black},
y grid style={lightgray!92.02614379084967!black},
ymajorgrids,
ymin=4.22346644208423e-05, ymax=4.21882010318113,
ymode=log,
ytick style={color=black},
ytick={1e-06,1e-05,0.0001,0.001,0.01,0.1,1,10,100},
yticklabels={\(\displaystyle {10^{-6}}\),\(\displaystyle {10^{-5}}\),\(\displaystyle {10^{-4}}\),\(\displaystyle {10^{-3}}\),\(\displaystyle {10^{-2}}\),\(\displaystyle {10^{-1}}\),\(\displaystyle {10^{0}}\),\(\displaystyle {10^{1}}\),\(\displaystyle {10^{2}}\)}
]
\addplot [semithick, color0]
table {%
0 2.5
1 1.8
2 1.5475
3 1.3275
4 1.13025
5 1.017225
6 0.91348975
7 0.8183844875
8 0.7203516825
9 0.64831651425
10 0.57761779595625
11 0.519856016360625
12 0.467870414724563
13 0.431047979373783
14 0.380326132006621
15 0.342293518805958
16 0.308064166925363
17 0.283601099927905
18 0.250432971657682
19 0.218678253557013
20 0.196810428201312
21 0.176919344568797
22 0.159227410111918
23 0.143595528341469
24 0.128153983788064
25 0.114547142383142
26 0.102202693946502
27 0.0911231294178056
28 0.0811683727835757
29 0.0722532435785682
30 0.062788178687895
31 0.0540586528610893
32 0.0486527875749804
33 0.0437875088174823
34 0.0401023681867097
35 0.0360921313680387
36 0.0324829182312349
37 0.0301571939105652
38 0.0271414745195087
39 0.0246422515964362
40 0.0221780264367925
41 0.0199903992836564
42 0.0175239307931122
43 0.0152939178788643
44 0.0133758312974695
45 0.0117210294299614
46 0.0102850381937612
47 0.0092565343743851
48 0.00833088093694659
49 0.00757349420736636
};
\addplot [semithick, color1]
table {%
0 0.75
1 1.75
2 1.4275
3 1.43875
4 1.22
5 1.20537875
6 1.044792
7 0.9209687625
8 0.79139389375
9 0.780359124375
10 0.63853806946875
11 0.629937922844063
12 0.618722629246641
13 0.484478091594024
14 0.418352896990809
15 0.412809633232486
16 0.405566563850657
17 0.31863641641716
18 0.275587736132185
19 0.239247270099514
20 0.236263044190864
21 0.196460360180673
22 0.193851723563038
23 0.159712175513157
24 0.141408899179516
25 0.126847751735776
26 0.112558593262038
27 0.0999773981393793
28 0.0886817646998741
29 0.0786795915335857
30 0.068520342690541
31 0.0587036654245346
32 0.0579527879937266
33 0.0569588760201441
34 0.0449435259024416
35 0.0443323314389586
36 0.0435361269098892
37 0.0340205307353246
38 0.0335549694954855
39 0.0274996820676392
40 0.0271217760020106
41 0.022228318921031
42 0.0192113516800254
43 0.0167278977348844
44 0.014646093223017
45 0.0128471334012009
46 0.0112812156079547
47 0.0111442402924404
48 0.0109599622437861
49 0.00923520401888664
};
\addplot [semithick, color2]
table {%
0 1
1 1.075
2 1.33
3 1.1125
4 1.017325
5 0.9970975
6 0.875926875
7 0.7644033
8 0.6417407175
9 0.6287596246875
10 0.5274554104875
11 0.515713957379063
12 0.501085768463484
13 0.405657814607836
14 0.345514386808531
15 0.338464360022442
16 0.329208665908957
17 0.267574372336465
18 0.228061568541345
19 0.200131768469026
20 0.196110292160034
21 0.163828643439841
22 0.160196940914396
23 0.132775611148887
24 0.120239553526906
25 0.105649800399881
26 0.0936097426060422
27 0.0826986924933131
28 0.0731854535170763
29 0.0647874418025137
30 0.0546191521611742
31 0.04833112937182
32 0.0473745489408223
33 0.046099542987401
34 0.0375607943485121
35 0.0366484409284908
36 0.0355581432394144
37 0.0285726427633016
38 0.0278578896157698
39 0.0227989912262691
40 0.0222786114052695
41 0.0184323040763944
42 0.0159318646766566
43 0.0139982685671697
44 0.0123184355049324
45 0.0108414017859479
46 0.00953788671186334
47 0.00934638523740914
48 0.00908789457809089
49 0.00772572909679354
};
\addplot [semithick, color3]
table {%
0 1.5
1 0.8
2 0.5075
3 0.469
4 0.4101
5 0.24606
6 0.22273975
7 0.1989625625
8 0.232465685
9 0.139479411
10 0.1815436893375
11 0.1089262136025
12 0.0653557281615
13 0.126359444957297
14 0.131675340904486
15 0.0790052045426916
16 0.0474031227256149
17 0.0848029822720374
18 0.0874360673761078
19 0.0805867406587121
20 0.0483520443952273
21 0.0585628761900978
22 0.0351377257140587
23 0.0460944784685026
24 0.03632480034914
25 0.0301863555272607
26 0.0256581578036184
27 0.0221900531624134
28 0.019364895830538
29 0.0170009932669766
30 0.0197578856672934
31 0.0188269846991522
32 0.0112961908194913
33 0.00677771449169479
34 0.0119309943941581
35 0.00715859663649487
36 0.00429515798189692
37 0.00882969184064165
38 0.00529781510438499
39 0.00769070396490795
40 0.00461442237894477
41 0.00630149123263181
42 0.00620705730833161
43 0.00566919585866396
44 0.00506733205689953
45 0.00449061992374629
46 0.0039645283266881
47 0.00237871699601286
48 0.00142723019760772
49 0.00162240455661295
};

\nextgroupplot[
log basis y={10},
tick align=outside,
tick pos=both,
title={CD},
x grid style={lightgray!92.02614379084967!black},
xmajorgrids,
xmin=-2.45, xmax=51.45,
xtick style={color=black},
y grid style={lightgray!92.02614379084967!black},
ymajorgrids,
ymin=4.22346644208423e-05, ymax=4.21882010318113,
ymode=log,
ytick style={color=black},
ytick={1e-06,1e-05,0.0001,0.001,0.01,0.1,1,10,100},
yticklabels={\(\displaystyle {10^{-6}}\),\(\displaystyle {10^{-5}}\),\(\displaystyle {10^{-4}}\),\(\displaystyle {10^{-3}}\),\(\displaystyle {10^{-2}}\),\(\displaystyle {10^{-1}}\),\(\displaystyle {10^{0}}\),\(\displaystyle {10^{1}}\),\(\displaystyle {10^{2}}\)}
]
\addplot [semithick, color0]
table {%
0 2.5
1 1.8
2 1.4025
3 1.2
4 1.040925
5 0.9368325
6 0.84314925
7 0.758834325
8 0.6829508925
9 0.61465580325
10 0.553190222925
11 0.4978712006325
12 0.44808408056925
13 0.403275672512325
14 0.362948105261093
15 0.326653294734983
16 0.293987965261485
17 0.264589168735336
18 0.238130251861803
19 0.214317226675623
20 0.19288550400806
21 0.173596953607254
22 0.260759810581564
23 0.276377079020547
24 0.254495339345925
25 0.22884188479422
26 0.205957696314798
27 0.185361926683318
28 0.166825734014987
29 0.150143160613488
30 0.135128844552139
31 0.121615960096925
32 0.127272478149174
33 0.118559973152394
34 0.105048593744828
35 0.0945437343703449
36 0.0850893609333104
37 0.0765804248399793
38 0.0689223823559814
39 0.0620301441203833
40 0.0558271297083449
41 0.0502444167375105
42 0.0452199750637594
43 0.0406979775573835
44 0.0366281798016451
45 0.0329653618214806
46 0.0296688256393326
47 0.0267019430753993
48 0.0240317487678594
49 0.0216285738910734
};
\addplot [semithick, color1]
table {%
0 0.75
1 1.75
2 1.4375
3 1.29875
4 1.134875
5 1.121105
6 1.1025653125
7 1.079918
8 1.05388654515625
9 1.02517537728125
10 0.994430322492578
11 0.962222280492032
12 0.9290437097689
13 0.895311840264304
14 0.861375109343921
15 0.827520841309568
16 0.793983086864368
17 0.760950060979803
18 0.728570917509518
19 0.696961767509252
20 0.666210939767704
21 0.636383528802138
22 0.34780120749898
23 0.330248559743463
24 0.285970359399602
25 0.254155471650738
26 0.251164888717237
27 0.247100964303151
28 0.242106912888559
29 0.236342909393847
30 0.229966601672953
31 0.223123401828709
32 0.151312816866827
33 0.133963716179404
34 0.115801002998516
35 0.114393317770133
36 0.112498495680336
37 0.110184588294462
38 0.107525663769704
39 0.104593703798615
40 0.101454631780027
41 0.0981666844485997
42 0.094780074682268
43 0.0913373314381342
44 0.0878739631040877
45 0.0844192445438911
46 0.0809970186281491
47 0.0776264557548939
48 0.0743227451326292
49 0.071097708606286
};
\addplot [semithick, color2]
table {%
0 1
1 1.075
2 1.195
3 1.06975
4 0.953425
5 0.93458375
6 0.909035125
7 0.880372578125
8 0.8505988403125
9 0.820793855539063
10 0.791509582049219
11 0.763001945888926
12 0.735365948876635
13 0.708613389608616
14 0.68271675273969
15 0.657633298980149
16 0.633317690775952
17 0.609728083818578
18 0.586828583151913
19 0.564589753511919
20 0.542988156216022
21 0.522005461390705
22 0.399737308905362
23 0.295563639778656
24 0.25381560395555
25 0.21965739983724
26 0.215093592837951
27 0.208958973631156
28 0.202108021867881
29 0.195016503756331
30 0.187940361908877
31 0.181009673636728
32 0.140654954209634
33 0.113042032844637
34 0.0983997076845481
35 0.0963464593739903
36 0.0936294827098806
37 0.090610261253779
38 0.087489745761646
39 0.0843758231050796
40 0.0813231872231703
41 0.0783567568330553
42 0.0754853094181002
43 0.0727093186379558
44 0.0700253752025983
45 0.0674286089514617
46 0.0649139541518613
47 0.0624767559890934
48 0.0601130109705353
49 0.0578194118056771
};
\addplot [semithick, color3]
table {%
0 1.5
1 0.8
2 0.535
3 0.453
4 0.39515
5 0.23709
6 0.142254
7 0.0853524
8 0.05121144
9 0.030726864
10 0.0184361184
11 0.01106167104
12 0.00663700262399999
13 0.0039822015744
14 0.00238932094464
15 0.001433592566784
16 0.000860155540070399
17 0.00051609332404224
18 0.000309655994425344
19 0.000185793596655206
20 0.000111476157993124
21 6.68856947958742e-05
22 0.104427846556059
23 0.121718600403496
24 0.10780016811713
25 0.0938831880379619
26 0.0563299128227772
27 0.0337979476936663
28 0.0202787686161998
29 0.0121672611697199
30 0.00730035670183192
31 0.00438021402109915
32 0.0379540203357853
33 0.0433125989762408
34 0.0399334461594237
35 0.0239600676956542
36 0.0143760406173925
37 0.00862562437043553
38 0.00517537462226132
39 0.00310522477335679
40 0.00186313486401407
41 0.00111788091840844
42 0.000670728551045066
43 0.00040243713062704
44 0.000241462278376224
45 0.000144877367025734
46 8.69264202154406e-05
47 5.21558521292643e-05
48 3.12935112775586e-05
49 1.87761067665352e-05
};

\nextgroupplot[
log basis y={10},
tick align=outside,
tick pos=both,
title={Mode independent LP},
x grid style={lightgray!92.02614379084967!black},
xlabel={Number of time-steps},
xmajorgrids,
xmin=-2.45, xmax=51.45,
xtick style={color=black},
y grid style={lightgray!92.02614379084967!black},
ymajorgrids,
ymin=4.22346644208423e-05, ymax=4.21882010318113,
ymode=log,
ytick style={color=black},
ytick={1e-06,1e-05,0.0001,0.001,0.01,0.1,1,10,100},
yticklabels={\(\displaystyle {10^{-6}}\),\(\displaystyle {10^{-5}}\),\(\displaystyle {10^{-4}}\),\(\displaystyle {10^{-3}}\),\(\displaystyle {10^{-2}}\),\(\displaystyle {10^{-1}}\),\(\displaystyle {10^{0}}\),\(\displaystyle {10^{1}}\),\(\displaystyle {10^{2}}\)}
]
\addplot [semithick, color0]
table {%
0 2.5
1 1.8
2 1.4025
3 1.2415
4 1.080375
5 0.9723375
6 0.864661125
7 0.75646605
8 0.65994139875
9 0.577229881125
10 0.5058699601125
11 0.44392812123
12 0.3898788874005
13 0.342567339886988
14 0.301075248086243
15 0.264647547418452
16 0.232646499387044
17 0.204524479706011
18 0.179806491683709
19 0.158078125464723
20 0.138976619628814
21 0.12218382203945
22 0.107420405619234
23 0.0944409827650323
24 0.0830299086163181
25 0.0729976400177261
26 0.0641775583965195
27 0.0564231883495413
28 0.0496057581331505
29 0.0436120577885836
30 0.0383425574821704
31 0.0345083017339533
32 0.031057471560558
33 0.0288321543480307
34 0.0255422548211195
35 0.0223546433296817
36 0.0195749186309017
37 0.0176174267678115
38 0.0158556840910304
39 0.0142701156819273
40 0.0135608101945991
41 0.0121166450407351
42 0.0109049805366616
43 0.00981448248299542
44 0.00883303423469588
45 0.00794973081122629
46 0.00715475773010366
47 0.0064392819570933
48 0.00579535376138397
49 0.00521581838524557
};
\addplot [semithick, color1]
table {%
0 0.75
1 1.75
2 1.4375
3 1.30575
4 1.179875
5 1.1636625
6 0.954933125
7 0.8284461875
8 0.72166930625
9 0.6320791025
10 0.554501659625
11 0.48694522199375
12 0.4278340850525
13 0.376006122778062
14 0.330508410200244
15 0.290541621654299
16 0.255420391935846
17 0.224550883149831
18 0.19741524628651
19 0.173560296599351
20 0.152588636648995
21 0.134151393584676
22 0.117942098345665
23 0.103691435132104
24 0.0911626896446807
25 0.0801477790547177
26 0.0704637773186694
27 0.0619498680509757
28 0.0544646693804802
29 0.0478838839143975
30 0.0420982335759921
31 0.0415895746609923
32 0.0409041423233547
33 0.0325248547997722
34 0.0281543475278505
35 0.0244866114722431
36 0.0214475427199953
37 0.0211843703102529
38 0.0208315612275423
39 0.020401582433337
40 0.0154662924833682
41 0.0134140033039256
42 0.0132353956283263
43 0.0130025128670901
44 0.0127233062333024
45 0.0124062411052172
46 0.0120594556212681
47 0.0116903552749811
48 0.0113054539902888
49 0.0109103502436706
};
\addplot [semithick, color2]
table {%
0 1
1 1.075
2 1.195
3 1.11825
4 0.955125
5 0.93752875
6 0.79223125
7 0.6868725625
8 0.60421680625
9 0.531773438125
10 0.468020483875
11 0.411740048875
12 0.362133358155625
13 0.3184494805135
14 0.280007995222356
15 0.246193236453241
16 0.216455266912617
17 0.190306012867008
18 0.167314101514992
19 0.147099147833266
20 0.129326167540538
21 0.11370037465188
22 0.0999624646449735
23 0.0878843965243292
24 0.0772656494931173
25 0.0679299146196127
26 0.059722177705432
27 0.0525061504114191
28 0.0461620096956444
29 0.0405844093634006
30 0.0356807313698177
31 0.0349622008352787
32 0.0339922380389864
33 0.0274765407076703
34 0.0234309681233667
35 0.0205404421103088
36 0.0180567186142413
37 0.017695132198483
38 0.017207862479947
39 0.0166624692848092
40 0.013267196413252
41 0.0111392080551307
42 0.0109081554952677
43 0.0106064105754296
44 0.0102722756702045
45 0.00992704172022719
46 0.00958214658669214
47 0.00924340289902242
48 0.00891347928683089
49 0.00859334340603923
};
\addplot [semithick, color3]
table {%
0 1.5
1 0.8
2 0.535
3 0.387
4 0.37845
5 0.22707
6 0.27833375
7 0.26977975
8 0.2452864125
9 0.21895792625
10 0.193937858125
11 0.171179240025
12 0.150819705785
13 0.132754553945125
14 0.11679171768075
15 0.102718286116771
16 0.0903259617373566
17 0.0794214380774661
18 0.069829777804388
19 0.0613947314247536
20 0.0539777221365547
21 0.0474563223627294
22 0.0417226038754678
23 0.0366815344791818
24 0.0322494930965386
25 0.0283529271372389
26 0.0249271537788181
27 0.0219152970526136
28 0.0192673489033293
29 0.0169393415004606
30 0.0148926184728644
31 0.00893557108371861
32 0.00536134265023117
33 0.00894298466788974
34 0.00907250200868996
35 0.00831519442814932
36 0.00743416619332149
37 0.0044604997159929
38 0.00267629982959574
39 0.00160577989775744
40 0.00397480581606481
41 0.00424336160907632
42 0.00254601696544579
43 0.00152761017926748
44 0.000916566107560486
45 0.000549939664536291
46 0.000329963798721775
47 0.000197978279233065
48 0.000118786967539839
49 7.12721805239034e-05
};

\nextgroupplot[
legend cell align={left},
legend style={at={(-0.05,1.7)}, draw=white!80.0!black,nodes={scale=0.8, transform shape}},
log basis y={10},
tick align=outside,
tick pos=both,
title={Mode dependent LP},
x grid style={lightgray!92.02614379084967!black},
xlabel={Number of time-steps},
xmajorgrids,
xmin=-2.45, xmax=51.45,
xtick style={color=black},
y grid style={lightgray!92.02614379084967!black},
ymajorgrids,
ymin=4.22346644208423e-05, ymax=4.21882010318113,
ymode=log,
ytick style={color=black},
ytick={1e-06,1e-05,0.0001,0.001,0.01,0.1,1,10,100},
yticklabels={\(\displaystyle {10^{-6}}\),\(\displaystyle {10^{-5}}\),\(\displaystyle {10^{-4}}\),\(\displaystyle {10^{-3}}\),\(\displaystyle {10^{-2}}\),\(\displaystyle {10^{-1}}\),\(\displaystyle {10^{0}}\),\(\displaystyle {10^{1}}\),\(\displaystyle {10^{2}}\)}
]
\addplot [semithick, color0]
table {%
0 2.5
1 1.8
2 1.4025
3 1.2415
4 1.105025
5 0.9832825
6 0.855915
7 0.7703235
8 0.69329115
9 0.623962035
10 0.5615658315
11 0.50540924835
12 0.454868323515
13 0.461453040506827
14 0.426515962825662
15 0.383864366543096
16 0.346368380623306
17 0.311363635377318
18 0.275958396127989
19 0.240221711776475
20 0.209484585055969
21 0.188536126550372
22 0.169456384427942
23 0.152510745985148
24 0.137559319403487
25 0.120704169895953
26 0.108633752906358
27 0.0975279367483503
28 0.0877751430735153
29 0.0790944745022066
30 0.0693628104945276
31 0.0624265294450748
32 0.0560242067164706
33 0.0504217860448235
34 0.0454272675745539
35 0.0398347020971813
36 0.0347832295122281
37 0.0313049065610053
38 0.0281744159049048
39 0.0260174927506463
40 0.0232772104361682
41 0.0208162938308753
42 0.0182818161531934
43 0.0162545704965023
44 0.0142817435150881
45 0.0126977915417299
46 0.0111723040944284
47 0.0100550736849855
48 0.00904956631648697
49 0.00829018818776108
};
\addlegendentry{$x_0$}
\addplot [semithick, color1]
table {%
0 0.75
1 1.75
2 1.4375
3 1.30575
4 1.192725
5 1.066365
6 0.93488375
7 0.921695175
8 0.904827854375
9 0.8848383845625
10 0.862310860104688
11 0.837802708834219
12 0.811818950075059
13 0.541654239087211
14 0.480133417789357
15 0.47351012506714
16 0.406245771920502
17 0.356504322964457
18 0.304248842524657
19 0.262390321563463
20 0.229030119205044
21 0.226181128135991
22 0.188159244104717
23 0.185665532317012
24 0.153010159346571
25 0.132393956791511
26 0.130692524707096
27 0.108229684794072
28 0.106788152224123
29 0.0879365526335866
30 0.0760574469064427
31 0.0750784890781427
32 0.0621604131382469
33 0.0613320877066616
34 0.0505012196018448
35 0.0436775363746918
36 0.0380546981384562
37 0.0375818942695464
38 0.0369505681214773
39 0.0292752799652914
40 0.025976384996508
41 0.0231510675398559
42 0.020061546673218
43 0.0176424420503395
44 0.015675620610229
45 0.0137816308248693
46 0.0122716110531823
47 0.0121023548472012
48 0.0118845459561963
49 0.00929219715469165
};
\addlegendentry{$x_1$}
\addplot [semithick, color2]
table {%
0 1
1 1.075
2 1.195
3 1.11825
4 0.992625
5 0.886485
6 0.750081
7 0.7357204375
8 0.7165751625
9 0.69516424496875
10 0.6728928633
11 0.650520994135234
12 0.628440154659422
13 0.498813795071886
14 0.401285018308723
15 0.392768873888291
16 0.34268930878673
17 0.297134049671996
18 0.248822246998937
19 0.217838526633612
20 0.191595831973076
21 0.187782361659063
22 0.156987868027375
23 0.153509572707994
24 0.127228594179536
25 0.110055101998759
26 0.107825616681786
27 0.09010078144082
28 0.0880979621936579
29 0.0730204193106788
30 0.0631729267523585
31 0.061894296831419
32 0.0517301109283175
33 0.0505806531588661
34 0.0419278051748297
35 0.0362746941797312
36 0.0318700893864882
37 0.0312323590167053
38 0.030376144025678
39 0.0246321623565117
40 0.0219452548194076
41 0.0193067995155607
42 0.0162911757394758
43 0.0150966707239557
44 0.0127239672965942
45 0.0118454305646347
46 0.0099857938923132
47 0.00979345818932057
48 0.00953482799289632
49 0.00779544735254933
};
\addlegendentry{$x_2$}
\addplot [semithick, color3]
table {%
0 1.5
1 0.8
2 0.535
3 0.387
4 0.305325
5 0.251925
6 0.278067
7 0.1668402
8 0.10010412
9 0.060062472
10 0.0360374832
11 0.02162248992
12 0.012973493952
13 0.130877428512684
14 0.152113730419451
15 0.0912682382516706
16 0.0849110065146644
17 0.0767244341360052
18 0.0901165835888013
19 0.085811082835308
20 0.0778921384608457
21 0.0467352830765074
22 0.0562505855624156
23 0.0337503513374494
24 0.0442020550765786
25 0.0431265408665386
26 0.0258759245199232
27 0.0319154931443264
28 0.0191492958865959
29 0.0252793107933699
30 0.0247158081794837
31 0.0148294849076902
32 0.0183106533293599
33 0.0109863919976159
34 0.0145106874308196
35 0.0141898360072938
36 0.0129308732388637
37 0.00775852394331825
38 0.00465511436599095
39 0.00793727455153199
40 0.00643185351141716
41 0.00541045223764935
42 0.00602554079817188
43 0.00464188797303352
44 0.00487608933400456
45 0.0037104413966617
46 0.00385326267159162
47 0.00231195760295497
48 0.00138717456177298
49 0.00246183542328846
};
\addlegendentry{$x_3$}
\end{groupplot}

\end{tikzpicture}
    \caption{Log scale states evolution of the system w.r.t. the number of time steps for policies generated from the four methods explained in this sequel.}
    \label{fig:numstate}
\end{figure}
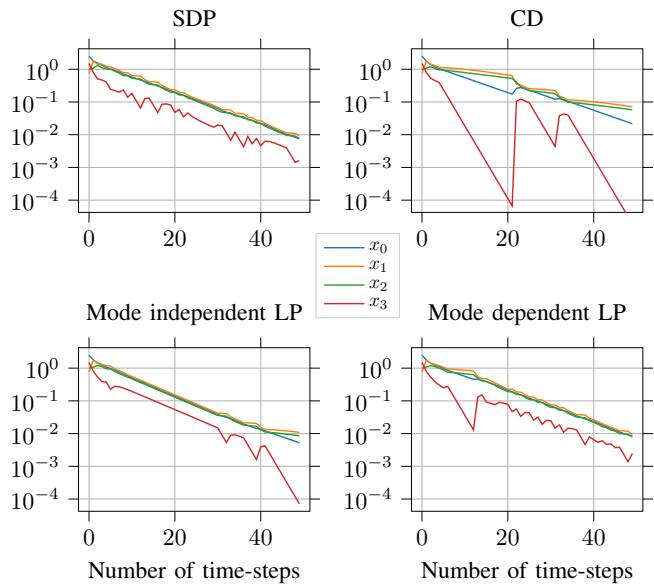

%% file: main.bbl
\begin{thebibliography}{10}
\providecommand{\url}[1]{#1}
\csname url@samestyle\endcsname
\providecommand{\newblock}{\relax}
\providecommand{\bibinfo}[2]{#2}
\providecommand{\BIBentrySTDinterwordspacing}{\spaceskip=0pt\relax}
\providecommand{\BIBentryALTinterwordstretchfactor}{4}
\providecommand{\BIBentryALTinterwordspacing}{\spaceskip=\fontdimen2\font plus
\BIBentryALTinterwordstretchfactor\fontdimen3\font minus
  \fontdimen4\font\relax}
\providecommand{\BIBforeignlanguage}[2]{{%
\expandafter\ifx\csname l@#1\endcsname\relax
\typeout{** WARNING: IEEEtran.bst: No hyphenation pattern has been}%
\typeout{** loaded for the language `#1'. Using the pattern for}%
\typeout{** the default language instead.}%
\else
\language=\csname l@#1\endcsname
\fi
#2}}
\providecommand{\BIBdecl}{\relax}
\BIBdecl

\bibitem{liberzon2003switching}
D.~Liberzon, \emph{Switching in systems and control}.\hskip 1em plus 0.5em
  minus 0.4em\relax Springer Science \& Business Media, 2003.

\bibitem{sun2006switched}
Z.~Sun, \emph{Switched linear systems: control and design}.\hskip 1em plus
  0.5em minus 0.4em\relax Springer Science \& Business Media, 2006.

\bibitem{zegers2018distributed}
F.~M. Zegers, H.-Y. Chen, P.~Deptula, and W.~E. Dixon, ``Distributed
  coordination of a multi-agent system with intermittent communication: A
  switched systems approach,'' in \emph{ASME 2018 Dynamic Systems and Control
  Conference}.\hskip 1em plus 0.5em minus 0.4em\relax American Society of
  Mechanical Engineers, 2018.

\bibitem{xucontroller}
Z.~Xu, F.~M. Zegers, B.~Wu, W.~Dixon, and U.~Topcu, ``Controller synthesis for
  multi-agent systems with intermittent communication: A metric temporal logic
  approach,'' in \emph{57th Annual Allerton Conference on Communication,
  Control, and Computing}, 2019, pp. 1015--1022.

\bibitem{wu2014stability}
B.~Wu, H.~Lin, and M.~Lemmon, ``Stability analysis for wireless networked
  control system in unslotted ieee 802.15. 4 protocol,'' in \emph{11th IEEE
  International Conference on Control \& Automation (ICCA)}.\hskip 1em plus
  0.5em minus 0.4em\relax IEEE, 2014, pp. 1084--1089.

\bibitem{zhang2017energy}
D.~Zhang, P.~Shi, W.-A. Zhang, and L.~Yu, ``Energy-efficient distributed
  filtering in sensor networks: A unified switched system approach,''
  \emph{IEEE Transactions on Cybernetics}, vol.~47, no.~7, pp. 1618--1629,
  2017.

\bibitem{wu2017formal}
B.~Wu, M.~D. Lemmon, and H.~Lin, ``Formal methods for stability analysis of
  networked control systems with ieee 802.15. 4 protocol,'' \emph{IEEE
  Transactions on Control Systems Technology}, vol.~26, no.~5, pp. 1635--1645,
  2017.

\bibitem{cetinkaya2018analysis}
A.~Cetinkaya, H.~Ishii, and T.~Hayakawa, ``Analysis of stochastic switched
  systems with application to networked control under jamming attacks,''
  \emph{IEEE Transactions on Automatic Control}, 2018.

\bibitem{wu2018privacy}
B.~Wu and H.~Lin, ``Privacy verification and enforcement via belief
  abstraction,'' \emph{IEEE control systems letters}, vol.~2, no.~4, pp.
  815--820, 2018.

\bibitem{wu2020privacy}
B.~Wu, H.~Lin, and U.~Topcu, ``Privacy verification and enforcement via belief
  manipulation,'' in \emph{Privacy in Dynamical Systems}.\hskip 1em plus 0.5em
  minus 0.4em\relax Springer, 2020, pp. 83--101.

\bibitem{zhang2008l_}
L.~Zhang and P.~Shi, ``$ l_2-l_\infty$ model reduction for switched lpv systems
  with average dwell time,'' \emph{IEEE Transactions on Automatic Control},
  vol.~53, no.~10, pp. 2443--2448, 2008.

\bibitem{puterman2014markov}
M.~L. Puterman, \emph{Markov decision processes: discrete stochastic dynamic
  programming}.\hskip 1em plus 0.5em minus 0.4em\relax John Wiley \& Sons,
  2014.

\bibitem{costa2006discrete}
O.~L.~V. Costa, M.~D. Fragoso, and R.~P. Marques, \emph{Discrete-time Markov
  jump linear systems}.\hskip 1em plus 0.5em minus 0.4em\relax Springer Science
  \& Business Media, 2006.

\bibitem{bolzern2015positive}
P.~Bolzern, P.~Colaneri \emph{et~al.}, ``Positive markov jump linear systems,''
  \emph{Foundations and Trends{\textregistered} in Systems and Control},
  vol.~2, no. 3-4, pp. 275--427, 2015.

\bibitem{shi2015survey}
P.~Shi and F.~Li, ``A survey on markovian jump systems: modeling and design,''
  \emph{International Journal of Control, Automation and Systems}, vol.~13,
  no.~1, pp. 1--16, 2015.

\bibitem{saravanakumar2017stability}
R.~Saravanakumar, M.~S. Ali, C.~K. Ahn, H.~R. Karimi, and P.~Shi, ``Stability
  of markovian jump generalized neural networks with interval time-varying
  delays,'' \emph{IEEE transactions on neural networks and learning systems},
  vol.~28, no.~8, pp. 1840--1850, 2017.

\bibitem{razaviyayn2013unified}
M.~Razaviyayn, M.~Hong, and Z.-Q. Luo, ``A unified convergence analysis of
  block successive minimization methods for nonsmooth optimization,''
  \emph{SIAM Journal on Optimization}, vol.~23, no.~2, pp. 1126--1153, 2013.

\bibitem{shen2017disciplined}
X.~Shen, S.~Diamond, M.~Udell, Y.~Gu, and S.~Boyd, ``Disciplined multi-convex
  programming,'' in \emph{2017 29th Chinese Control And Decision Conference
  (CCDC)}.\hskip 1em plus 0.5em minus 0.4em\relax IEEE, 2017, pp. 895--900.

\bibitem{hespanha1999stability}
J.~P. Hespanha and A.~S. Morse, ``Stability of switched systems with average
  dwell-time,'' in \emph{Proceedings of the 38th IEEE Conference on Decision
  and Control}, vol.~3.\hskip 1em plus 0.5em minus 0.4em\relax IEEE, 1999, pp.
  2655--2660.

\bibitem{zhao2011stability}
X.~Zhao, L.~Zhang, P.~Shi, and M.~Liu, ``Stability and stabilization of
  switched linear systems with mode-dependent average dwell time,'' \emph{IEEE
  Transactions on Automatic Control}, vol.~57, no.~7, pp. 1809--1815, 2011.

\bibitem{narendra1994common}
K.~S. Narendra and J.~Balakrishnan, ``A common lyapunov function for stable lti
  systems with commuting a-matrices,'' \emph{IEEE Transactions on automatic
  control}, vol.~39, no.~12, pp. 2469--2471, 1994.

\bibitem{branicky1998multiple}
M.~S. Branicky, ``Multiple lyapunov functions and other analysis tools for
  switched and hybrid systems,'' \emph{IEEE Transactions on automatic control},
  vol.~43, no.~4, pp. 475--482, 1998.

\bibitem{morse1996supervisory}
A.~S. Morse, ``Supervisory control of families of linear set-point
  controllers-part i. exact matching,'' \emph{IEEE transactions on Automatic
  Control}, vol.~41, no.~10, pp. 1413--1431, 1996.

\bibitem{zhang2008exponential}
L.~Zhang, E.-K. Boukas, and P.~Shi, ``Exponential h∞ filtering for uncertain
  discrete-time switched linear systems with average dwell time: A
  $\mu$-dependent approach,'' \emph{International Journal of Robust and
  Nonlinear Control: IFAC-Affiliated Journal}, vol.~18, no.~11, pp. 1188--1207,
  2008.

\bibitem{weiss2007automata}
G.~Weiss and R.~Alur, ``Automata based interfaces for control and scheduling,''
  in \emph{International Workshop on Hybrid Systems: Computation and
  Control}.\hskip 1em plus 0.5em minus 0.4em\relax Springer, 2007, pp.
  601--613.

\bibitem{wang2016stability}
Y.~Wang, N.~Roohi, G.~E. Dullerud, and M.~Viswanathan, ``Stability analysis of
  switched linear systems defined by regular languages,'' \emph{IEEE
  Transactions on Automatic Control}, vol.~62, no.~5, pp. 2568--2575, 2016.

\bibitem{jha2010synthesizing}
S.~Jha, S.~Gulwani, S.~A. Seshia, and A.~Tiwari, ``Synthesizing switching logic
  for safety and dwell-time requirements,'' in \emph{Proceedings of the 1st
  ACM/IEEE International Conference on Cyber-Physical Systems}.\hskip 1em plus
  0.5em minus 0.4em\relax ACM, 2010, pp. 22--31.

\bibitem{koo2001mode}
T.~J. Koo, G.~J. Pappas, and S.~Sastry, ``Mode switching synthesis for
  reachability specifications,'' in \emph{International Workshop on Hybrid
  Systems: Computation and Control}.\hskip 1em plus 0.5em minus 0.4em\relax
  Springer, 2001, pp. 333--346.

\bibitem{liu2013synthesis}
J.~Liu, N.~Ozay, U.~Topcu, and R.~M. Murray, ``Synthesis of reactive switching
  protocols from temporal logic specifications,'' \emph{IEEE Transactions on
  Automatic Control}, vol.~58, no.~7, pp. 1771--1785, 2013.

\bibitem{zhang2008analysis}
L.~Zhang, E.-K. Boukas, and J.~Lam, ``Analysis and synthesis of markov jump
  linear systems with time-varying delays and partially known transition
  probabilities,'' \emph{IEEE Transactions on Automatic Control}, vol.~53,
  no.~10, pp. 2458--2464, 2008.

\bibitem{bolzern2010markov}
P.~Bolzern, P.~Colaneri, and G.~De~Nicolao, ``Markov jump linear systems with
  switching transition rates: mean square stability with dwell-time,''
  \emph{Automatica}, vol.~46, no.~6, pp. 1081--1088, 2010.

\bibitem{chen2012h}
L.~Chen, Y.~Leng, H.~Guo, P.~Shi, and L.~Zhang, ``H∞ control of a class of
  discrete-time markov jump linear systems with piecewise-constant tps subject
  to average dwell time switching,'' \emph{Journal of the Franklin Institute},
  vol. 349, no.~6, pp. 1989--2003, 2012.

\bibitem{lin2009stability}
H.~Lin and P.~J. Antsaklis, ``Stability and stabilizability of switched linear
  systems: a survey of recent results,'' \emph{IEEE Transactions on Automatic
  control}, vol.~54, no.~2, pp. 308--322, 2009.

\bibitem{gallager2013stochastic}
R.~G. Gallager, \emph{Stochastic processes: theory for applications}.\hskip 1em
  plus 0.5em minus 0.4em\relax Cambridge University Press, 2013.

\bibitem{Boyd}
S.~Boyd and L.~Vandenberghe, \emph{Convex Optimization}.\hskip 1em plus 0.5em
  minus 0.4em\relax Cambridge University Press, 2004.

\bibitem{nesterov1994interior}
Y.~Nesterov and A.~Nemirovskii, \emph{Interior-point polynomial algorithms in
  convex programming}.\hskip 1em plus 0.5em minus 0.4em\relax Siam, 1994,
  vol.~13.

\bibitem{vanantwerp2000tutorial}
J.~G. VanAntwerp and R.~D. Braatz, ``A tutorial on linear and bilinear matrix
  inequalities,'' \emph{Journal of process control}, vol.~10, no.~4, pp.
  363--385, 2000.

\bibitem{baier2008principles}
C.~Baier and J.-P. Katoen, \emph{Principles of model checking}.\hskip 1em plus
  0.5em minus 0.4em\relax MIT Press, 2008.

\bibitem{xu2013block}
Y.~Xu and W.~Yin, ``A block coordinate descent method for regularized
  multiconvex optimization with applications to nonnegative tensor
  factorization and completion,'' \emph{SIAM Journal on imaging sciences},
  vol.~6, no.~3, pp. 1758--1789, 2013.

\bibitem{bremaud2013markov}
P.~Br{\'e}maud, \emph{Markov chains: Gibbs fields, Monte Carlo simulation, and
  queues}.\hskip 1em plus 0.5em minus 0.4em\relax Springer Science \& Business
  Media, 2013, vol.~31.

\bibitem{brafman2002r}
R.~I. Brafman and M.~Tennenholtz, ``R-max-a general polynomial time algorithm
  for near-optimal reinforcement learning,'' \emph{Journal of Machine Learning
  Research}, vol.~3, no. Oct, pp. 213--231, 2002.

\bibitem{funderlic1986sensitivity}
R.~E. Funderlic and C.~Meyer~Jr, ``Sensitivity of the stationary distribution
  vector for an ergodic markov chain,'' \emph{Linear Algebra and its
  Applications}, vol.~76, pp. 1--17, 1986.

\bibitem{meyer1975role}
C.~D. Meyer, Jr, ``The role of the group generalized inverse in the theory of
  finite markov chains,'' \emph{Siam Review}, vol.~17, no.~3, pp. 443--464,
  1975.

\bibitem{mosek}
\BIBentryALTinterwordspacing
M.~ApS, \emph{The MOSEK optimization toolbox for MATLAB manual. Version 8.1.},
  2017. [Online]. Available: \url{http://docs.mosek.com/8.1/toolbox/index.html}
\BIBentrySTDinterwordspacing

\bibitem{gurobi}
\BIBentryALTinterwordspacing
L.~Gurobi~Optimization, ``Gurobi optimizer reference manual,'' 2019. [Online].
  Available: \url{http://www.gurobi.com}
\BIBentrySTDinterwordspacing

\bibitem{cvx}
I.~CVX~Research, ``{CVX}: Matlab software for disciplined convex programming,
  version 2.0,'' \url{http://cvxr.com/cvx}, Aug. 2012.

\bibitem{arantzerExample}
A.~{Rantzer}, ``Distributed control of positive systems,'' in \emph{2011 50th
  IEEE Conference on Decision and Control and European Control Conference}, Dec
  2011, pp. 6608--6611.

\bibitem{wu2019switched}
B.~Wu, M.~Cubuktepe, and U.~Topcu, ``Switched linear systems meet markov
  decision processes: Stability guaranteed policy synthesis,'' \emph{58th IEEE
  Conference on Decision and Control}, 2019, to appear, arXiv preprint
  arXiv:1904.11456.

\bibitem{rantzer2012optimizing}
A.~Rantzer, ``Optimizing positively dominated systems,'' in \emph{2012 IEEE
  51st IEEE Conference on Decision and Control (CDC)}.\hskip 1em plus 0.5em
  minus 0.4em\relax IEEE, 2012, pp. 272--277.

\end{thebibliography}
